\documentclass [12pt]{article}
\usepackage[numbers]{natbib}
\usepackage{graphicx}
\usepackage[small]{subfigure,epsfig}
\usepackage{indentfirst}
\usepackage{amsmath,latexsym,enumerate}
\usepackage{amssymb}
\usepackage{amsfonts}
\usepackage{array,longtable,lscape}
\usepackage{amsthm}

\usepackage{chngcntr}
\usepackage{apptools}
\usepackage{color}
\AtAppendix{\counterwithin{lemma}{section}}
\usepackage{authblk}

\theoremstyle{plain} 

\newtheorem{proposition}{Proposition}
\newtheorem{remark}{Remark}

\numberwithin{equation}{section} \numberwithin{lemma}{section} \numberwithin{theorem}{section} \numberwithin{proposition}{section} \numberwithin{corollary}{section}
\numberwithin{remark}{section}

\usepackage{geometry} 
\begin{document}

\title{Integrability of oscillators and transcendental invariant curves}

\author[1]{Jaume Gin\'e}

\author[2,3]{Dmitry I. Sinelshchikov}

\affil[1]{Departament de Matem\`atica, Universitat de Lleida, Avda. Jaume II, 69; 25001 Lleida, Catalonia, Spain}
\affil[2]{Instituto Biofisika (UPV/EHU, CSIC), University of the Basque Country, Leioa E-48940, Spain}
\affil[3]{Ikerbasque Foundation, Bilbao 48013, Spain}

\maketitle

\begin{abstract}
In this work we study the integrability of a family of nonlinear oscillators. Dynamical systems from this family appear in different applications from mechanics to chemistry. We propose an approach for finding first integrals and integrating factors, which is based on the construction and classification of transcendental invariant curves whose cofactors are polynomial or rational in one of the variables. We demonstrate that this approach can be efficiently used for finding non-Liouvillian and non-Puiseux integrable dynamical systems. Its application involves finding solutions only of linear algebraic and linear ordinary differential equations. This allows one to study singularities, including essential ones, of the invariant curves in the complex plane. We illustrate this approach by proving non-Liouvillian integrability of two dynamical systems from the Painlev\'e--Gambier classification and non-Puiseux integrability of an oscillator from the considered family. Furthermore, we construct equivalence classes of the first two dynamical systems with respect to nonlocal transformations. We show that among these equivalence classes there are interesting examples of integrable dynamical systems.
\end{abstract}

\section{Introduction}
Autonomous nonlinear oscillators form an important subclass of two-dimensional dynamical systems, as they are ubiquitous in different applications from physics and mechanics to biology and chemistry \cite{Andronov,Murray,Izhikevich,Jenkins}. In addition, some general two-dimensional dynamical systems, like the Lotka-Volterra system and its variants, can be transformed into nonlinear oscillators, such as the Li\'enard equations and their generalizations (see, e.g. \cite{Ghosh2013,Saha2017,Sinelshchikov2023a}). Consequently, understanding of the behavior of the  trajectories of nonlinear oscillators is an important problem. One of the possible solutions to this problem consists in establishing their integrability. Recall that an autonomous two-dimensional dynamical system or an autonomous nonlinear oscillator is called integrable when there exists a global autonomous first integral (see, e.g \cite{Llibre_book,Zhang}). Notice that throughout this work under a global first integral we understand a smooth function that is defined on whole domain of definition of a given dynamical system except, perhaps, for a set of zero Lebesgue measure.

Typically, the integrability of polynomial dynamical systems is considered in the relation to the existence of algebraic invariant curves and exponential factors (see, \cite{Llibre_book,Libre2010,Zhang} and references therein). This may lead only to the Liouvillian integrability, i.e. to the global first integrals that are Liouvillian functions. Let us recall that, roughly speaking, Liovillian functions are formed by algebraic operations and primitives of the elementary functions (for the complete definition of the set of Liouvillian functions see, \cite{Zhang,Kovacic1986}). For instance, Liouvillian integrability for several families of nonlinear oscillators has been studied in \cite{Gin,GV,GV1,GSi,Sinelshchikov2020,Sinelshchikov2021,Sinelshchikov2023a}.

On the other hand, it is known (see e.g. \cite{Gine2003,Gine2004,Gine2011,Gine2010,Demina2022}) that polynomial dynamical systems can be integrable in  non-Liouvillian sense.  That is they do not possess a Liouvillian global first integral, but there is a global first integral that is a non-Liouvillian function. Integrability of such dynamical system can be explained with the help of more general theories like the theory of generalized cofactors or the Weierstrass and Puiseux integrability \cite{Gine2003,Gine2010,Demina2022}, which includes Liouvillian integrability. Furthermore, there are non-polynomial, e.g. rational, nonlinear oscillators that appear in various applications {see for instance \cite{GV2}}. Therefore, it is interesting to consider the extension of integrability theories to not-necessarily polynomial nonlinear oscillators with non-Liouvillian first integrals.

It is important to remark that, when we deal with integrability theories that are connected to a certain functional class of a first integral or integrating factor, such theories are generally not invariant under an arbitrary change of coordinates. In other words, if one makes a change of variables, the new system may not be integrable according to the same theory of integrability. For instance, if we apply an arbitrary analytical change of variables to a Liouvillian integrable system, we can obtain a non-Liouvillian integrable system. The same can happen for any theory of integrability that is connected to a functional class of a first integral or integrating factor. In particular, below we show that algebraic transformations can map a Puiseux integrable system to a non-Puiseux integrable one.

Another interesting approach to integrability of polynomial dynamics systems is based on the applications of the differential Galois theory (see, \cite{Li2017,Szuminski2018,Szuminski2020,Li2020}). This theory requires the existence $n-k$ pairwise commuting generalized symmetries and $k$ common first integrals for integrability of an $n$ dimensional dynamical system ($n\in\mathbb{N},\, n>1,\,k\in\mathbb{N}\cup\{0\},\, n-k\geq0$). This type of integrability is call B(Bogoyavlenski)-integrability. However, one can show for any two dimensional dynamical system B--integrability is equivalent to integrability defined above.


In this work we use a combination of the methods discussed above to study the integrability of nonlinear oscillators. We propose an algorithm for the construction and classification of transcendental invariant curves, that are polynomials with respect to one of the variables, which we call the privileged variable. We assume that the cofactors of these invariant curves are rational or polynomial functions with respect to the privileged variable. If one is only interested in establishing the integrability of a dynamical system, one needs to start with zero degree, with respect to the privileged variable, invariant curves and consequently increase their order until one of the integrability correlations on the cofactors presented below is satisfied. On the other hand, if it is necessary to obtain classification of invariant curves, we can find an upper bound on the degree of an invariant curve with respect to the privileged variable using the technique from \cite{Demina2022}. We demonstrate that this can be used for proving that a studied dynamical system is not Liouvillian or Puiseux integrable. It is also worth noting that the application of the proposed algorithm involves finding solutions of only \textit{linear} algebraic and differential correlation on the coefficients of an invariant curve as a polynomial in the privileged variable. Therefore, working with linear equations often allows one to find explicit expressions for the transcendental invariant curves and the corresponding first integrals and integrating factors. Moreover, even if we do not know the explicit solutions of these linear ordinary differential equations, we can determine and classify singular points of their solutions and, hence, of the invariant curves and ultimately of the integrating factor and first integral.

We consider the following family of oscillators
\begin{equation}\label{eq:Lienard}
  x_{tt}+h(x)x_{t}^{2}+f(x)x_{t}+g(x)=0,
\end{equation}
where $h$, $f$, $g$ are arbitrary sufficiently smooth functions such that $f g \neq0$. Family \eqref{eq:Lienard} generalize the Li\'enard equation ($h=0$) and contains a lot of oscillators with different applications (see, e.g. \cite{Ghosh2013,Saha2017,Kazakov,Polyanin}).

We study integrability of two oscillators from \eqref{eq:Lienard}, one rational and one polynomial, that belong to the Painlev\'e--Gambier classification \cite{Ince}. Namely, we consider a rational oscillator
\begin{equation}\label{eq:Ince_XV}
  w_{\tau\tau}-\frac{w_{\tau}^{2}}{w}+\frac{w_{\tau}}{w}+2\alpha^{3}w^{2}=0,
\end{equation}
where $\alpha\neq0$ is an arbitrary parameter and a polynomial one
\begin{equation}\label{eq:Ince_V}
 w_{\tau\tau}+(2w-1) w_{\tau}-w^{2}+\beta^{2}=0,
\end{equation}
where $\beta$ is an arbitrary parameter.

We also consider a rational oscillator
\begin{equation}\label{eq:ie_1a}
w_{tt}-\frac{2w-1}{w(w-1)}w_{t}^{2}+\frac{(b-2)w^{2}-(2a-1)w+a}{w(w-1)}w_{t}+(b-1)w-a=0,
\end{equation}
where $a$ and $b$ are real parameters, constrains on which will be introduced below.

Both equations \eqref{eq:Ince_XV} and \eqref{eq:Ince_V} are of the Painlev\'e type. However, for equation \eqref{eq:Ince_XV} Liouvillian/non-Liouvillian integrability has not been studied, while an autonomous first integral was found in \cite{Sinelshchikov2019}. Liouvillian integrable cases of equation \eqref{eq:Ince_V}, that correspond to certain values of the parameter $\beta$ were considered in \cite{Demina2020}. On the other hand, the integrability in the case of an arbitrary $\beta$ has not been studied. As far as oscillator \eqref{eq:ie_1a} is concerned, to the best of our knowledge, its Liouvillian/Puiseux integrability has not been studied.

Here we prove that the dynamical systems corresponding to \eqref{eq:Ince_XV} and \eqref{eq:Ince_V} are integrable, but not in the Liouvillian sense and the dynamical system corresponding to \eqref{eq:ie_1a} is not globally Puiseux integrable. Using proposed algorithm, we explicitly construct transcendental invariant curves of these dynamical systems that yield their first integrals and integrating factors.

In order to extend families of non-Liouvillian integrable oscillators from \eqref{eq:Lienard} we use the approach to the integrability of dynamical systems, which is based on the applications of nonlocal transformations \cite{Sinelshchikov2020,Sinelshchikov2020a,Sinelshchikov2021,Sinelshchikov2021a,Sinelshchikov2023,GSi}. This approach does not depend on the functional class of both a first integral and an integrating factor and is based on construction of equivalence classes for integrable nonlinear oscillators and dynamical systems. It is easy to demonstrate that family of equations \eqref{eq:Lienard} is closed with respect to nonlocal transformations \cite{Duarte1994,Meleshko2010,Sinelshchikov2020a,Sinelshchikov2023a}
\begin{equation}\label{eq:GS}
  w=F(x), \quad d\tau=G(x)dt.
\end{equation}
Consequently, the construction of equivalence classes for integrable members of \eqref{eq:Lienard} is a non-trivial problem. In this work we obtain such equivalent classes for oscillators  \eqref{eq:Ince_XV} and \eqref{eq:Ince_V}. Furthermore, transformations \eqref{eq:GS} preserve first integrals, integrating factors and invariant curves for members of \eqref{eq:Lienard} \cite{Sinelshchikov2021,Sinelshchikov2021a,Sinelshchikov2023a}. This will be used to prove non-Liouvillian integrability of constructed equivalence classes and to obtain explicit expressions for invariant curves, first integrals and integrating factors.

Let us remark that while we are unaware of possible applications of directly equations \eqref{eq:Ince_XV} and \eqref{eq:Ince_V}, their equivalence classes with respect to nonlocal transformations contains applied dynamical systems, for example, a particular case of the cubic Kolmohorov system.

The rest of this work is organized as follows. In the next section we provide basic definitions and correlations on the cofactors of transcendental invariant curves that lead to the existence of a first integral or integrating factor. There we also introduce an algorithm for establishing integrability through transcendental invariant curves. In Sections 3 and 4 we apply the proposed algorithm to study the integrability of \eqref{eq:Ince_XV} and \eqref{eq:Ince_V}, respectively. In Section 5 we demonstrate that the dynamical system that corresponds to \eqref{eq:ie_1a} is not globally Puiseux integrable and construct its first integral and integrating factor. We obtain equivalence classes for \eqref{eq:Ince_XV} and \eqref{eq:Ince_V} with respect to transformations \eqref{eq:GS} in Section 6.  We give several examples of applied dynamical systems from the constructed equivalence classes in Section 7. In the last Section we briefly summarize and discuss the results of this manuscript.

\section{Basic definitions and preliminary results}

Let us introduce some basic definitions concerning integrability of two-dimensional dynamical systems. Consider the differential system
\begin{equation}\label{eq:gen_sys}
  x_{t}=P(x,y), \quad y_{t}=Q(x,y).
\end{equation}
Throughout this manuscript we assume that the functions $P$ and $Q$ are smooth everywhere in $\mathbb{C}^{2}$ except, perhaps, for a set of zero Lebesgue measure. We also denote by $\mathcal{X}=P\partial_{x}+Q\partial_{y}$ the vector filed associated to \eqref{eq:gen_sys}.

It is clear that \eqref{eq:Lienard} can be converted into the equivalent dynamical system
\begin{equation}\label{eq:Lienard_sys}
  x_{t}=y, \quad y_{t}=-hy^{2}-fy-g.
\end{equation}
Therefore, all definitions and results presented below for \eqref{eq:gen_sys} can be applied to studying integrability of \eqref{eq:Lienard}.

It is said that a smooth globally non-constant function $I(x,y)$ is a first integral of \eqref{eq:gen_sys} if it satisfies the equation
\begin{equation}\label{eq:I_def}
\mathcal{X}I= PI_{x}+QI_{y}=0.
\end{equation}
Let us recall that under an integrable two-dimensional dynamical system in this work we understand a dynamical system that possesses a global autonomous first integral. Non-autonomous first integrals of \eqref{eq:gen_sys}, i.e. first integrals that explicitly depend on time, are called invariants and satisfy the equation $I_{t}+\mathcal{X}I=0$. For complete integrability of \eqref{eq:gen_sys} it is required to have two global invariants.

The existence of a first integral leads to the existence of an integrating factor $M(x,y)$ (and vice versa), which is a smooth function satisfying
\begin{equation}\label{eq:M_def}
 (PM)_{x}+(QM)_{y}=0.
\end{equation}
For some classes of integrable dynamical systems, e.g. Liouvillian integrable ones, integrating factors can have simpler structure than first integrals \cite{Llibre_book,Zhang}.

An invariant curve of \eqref{eq:gen_sys} is a smooth function $H$ that satisfies
\begin{equation}\label{eq:L_def}
 PH_{x}+QH_{y}=\lambda H,
\end{equation}
where $\lambda(x,y)$ is the cofactor.

For polynomial dynamical systems an important role play exponential factors, which are defined by
\begin{equation}\label{eq:E_def}
 \mathcal{X}{\rm e}^{\tilde{f}/\tilde{g}}=L {\rm e}^{\tilde{f}/\tilde{g}},
\end{equation}
where $\tilde{f},\,\tilde{g},\,L,P,Q \in \mathbb{C}[x,y]$, where $\mathbb{C}[x,y]$ denotes the ring of polynomials in $x$ and $y$ with coefficients in $\mathbb{C}$.

Let us recall that a two-dimensional dynamical system is Liouvillian integrable if there exists a smooth Liouvillian function that satisfies \eqref{eq:I_def} in the whole domain of definition of the vector field  $\mathcal{X}$  associated to \eqref{eq:gen_sys}, except, perhaps for a set of zero Lebesgue measure. Now we can formulate necessary and sufficient condition for Liouvillian integrability of a two dimensional polynomial dynamical system \cite{Singer1992,Christopher1998,Llibre_book,Zhang}. Such system is Liouvillian integrable if and only if it has an integrating factor that is a Darboux function, i.e. of the form
\begin{equation}\label{eq:M_Darboux}
 M={\rm e}^{\tilde{f}/\tilde{g}}\prod\limits_{i}\tilde{f}_{i}^{l_{i}},
\end{equation}
where $\tilde{f},\,\tilde{g},\,\tilde{f}_{i}\in\mathbb{C}[x,y]$ with $l_{i} \in \mathbb{C}$.

Finally, let us briefly recall the definition of the Puiseux integrability for polynomial systems of the form \eqref{eq:gen_sys} (for the details see \cite{Demina2022}). We call a function of the form
\begin{equation}\label{eq:Puiseux_P}
\sum\limits_{j=0}^{l}a_{j}(x)y^{j},
\end{equation}
where $a_{j}(x)$ are Puiseux series around a point $x_{0}\in\mathbb{C}\cup\{\infty\}$ and $j\in\mathbb{N}\cup\{0\}$, a Puiseux polynomial. If a polynomial dynamical system possesses an integrating factor
\begin{equation}\label{eq:Puiseux_IF}
M=\exp\left\{\frac{\widetilde{A}}{\widetilde{B}}\right\}\prod\limits_{j}\widetilde{C}_{j}^{\alpha_{j}},
\end{equation}
where $\widetilde{A}$, $\widetilde{B}$ and $\widetilde{C}_{j}$ are Puiseux polynomials and $\alpha_{j}\in\mathbb{C}$, it is called Puiseux integrable near $x=x_{0}$.

Let us remark that in \eqref{eq:Puiseux_P} we privilege variable $y$ over $x$. However one can interchange $x$ and $y$ privileging variable $x$ over $y$, which does not impact the definition of Puiseux integrability.

In this work we are interested in non-Liovillian and non-Puiseux integrable systems for which one cannot establish integrability using only algebraic invariant curves, exponential factors and local Puiseux series. Notice also that we do not require any specific form of an integrating factor or a first integral, like in the definition of the Weierstrass or Puisex integrability.

Thus, we need to consider transcendental invariant curves with possibly transcendental cofactors. A notion of a generalized (transcendental) invariant curve with a polynomial cofactor for polynomial systems was introduced in \cite{Gine2003} and generalized in \cite{Gine2004,Gine2013}. Similar results can be obtained for an arbitrary smooth system \eqref{eq:gen_sys} and non-polynomial cofactors. Let us note that for a formal invariant curve of \eqref{eq:gen_sys} there is a formal cofactor verifying that $\mathcal{X}H=\lambda H$ (see, \cite{Seidenberg1968,CGGL2003}).

As a consequence, we can connect transcendental first integrals and integrating factors of \eqref{eq:gen_sys} with transcendental invariant curves:

\begin{proposition}
\label{p:p5}
Suppose that \eqref{eq:gen_sys} has $k\in\mathbb{N}$ invariant curves $H_{i}$, $i=\overline{1,k}$ with the associated cofactors $\lambda_{i}$, $i=\overline{1,k}$. Assume also that $\Omega\subseteq\mathbb{C}^{2}$ is a domain where $P$, $Q$ and $H_{i}$, $i=\overline{1,k}$ are smooth and $\lambda_{i}$, $i=\overline{1,k}$ are defined, $\Omega\neq\varnothing$ and $q_{j}\in\mathbb{C}$, $j=\overline{1,k}$. Then two following statements hold:

\begin{enumerate}
  \item[(I)] The function $I=H_{1}^{q_{1}}\ldots H_{k}^{q_{k}}$ is a first integral of \eqref{eq:gen_sys} if and only if $\sum\limits_{j=1}^{k}q_{j}\lambda_{j}=0$.
  \item[(II)] The function $M=H_{1}^{q_{1}}\ldots H_{k}^{q_{k}}$ is an integrating factor of \eqref{eq:gen_sys} if and only if $\sum\limits_{j=1}^{k}q_{j}\lambda_{j}=-\mbox{div}\mathcal{X}$.
\end{enumerate}
\end{proposition}

\begin{proof}
  The proof is by direct substitution of the expressions for a first integral and an integrating factor into their respective definitions.
\end{proof}

\begin{remark}
Generalized exponential factors are included into Proposition \ref{p:p5} since any formal invariant curve $H$ can be presented in the form $\exp\{\tilde{H}\}$, where $\tilde{H}=\ln H$.
\end{remark}

On the basis of Proposition \ref{p:p5} we propose the following algorithm for establishing integrability of \eqref{eq:gen_sys}:
\begin{enumerate}
  \item We look for a change of variables that simplifies the considered dynamical system. Often, one can use an algebraic invariant curve of the considered system as a candidate for a suitable transformation (see, \cite{GS}). One can also rescale the time with the help of nonlocal transformations of the form $d\tau=r(x,y)dt$. This step is optional and one can directly proceed to the next one.

  \item We consider the simplified system, privilege one of the variables and look for invariant curves that are polynomials in this variable with the coefficients that are arbitrary sufficiently smooth functions. The maximal degree of the invariant curve with respect to the privileged variable can be estimated on the basis of the number of Puiseux series around $\infty$ for the projection of \eqref{eq:gen_sys} on the $(x,y)$ plane (see, \cite{Demina2022}). On the other hand, one can start with invariant curves of zero and first degree, with respect to the privileged variable, and proceed with the rest of the steps of the algorithm. If the integrability is not established with the help of these invariant curves, one should consequently increase the order of an invariant curve by one.

  \item Third, we make an assumption on the form of the cofactor. For example, if the simplified system is polynomial we assume that the cofactor is a polynomial in the privileged variable of the degree less by one than the degree of the corresponding vector field. For a rational vector field, we assume that cofactors are rational functions with respect to the privileged variable.

\item Fourth, we substitute the proposed form of an invariant curve and the respective cofactor into \eqref{eq:L_def} and collect terms at the same powers of the privileged variable. As a result, we obtain a system of \textit{linear} ordinary differential and algebraic equations for the coefficients of an invariant curve and its cofactor with respect to the privileged variable. Solving this system of equations we find, generally, a set of transcendental invariant curves of \eqref{eq:gen_sys} and corresponding cofactors. We also look for generalized invariants of \eqref{eq:gen_sys} that depend only on one of the variables. Notice that at this step generalized exponential factors can be also constructed.

\item Fifth, we apply Proposition \ref{p:p5} to the cofactors of constructed invariant curves and check whether either $I$ or $II$ holds. If either $I$ or $II$ from Proposition \ref{p:p5} holds then we obtain that considered dynamical system is integrable and we stop the algorithm.
\end{enumerate}

\begin{remark}
  The proposed algorithm can be used to prove the absence of the algebraic invariant curves of polynomial systems \eqref{eq:gen_sys}, if one obtains the maximal degree of an invariant curve with respect to the privileged variable at Step 2. Then, the classification of all invariant curves that are polynomial in one variable is possible. If these curves are not algebraic and do not degenerate into algebraic ones, we obtain that the corresponding dynamical system do not have algebraic invariant curves.
\end{remark}

Below we use Proposition \ref{p:p5} and the algorithm above for constructing first integrals and integrating factors for dynamical systems \eqref{eq:gen_sys} with transcendental invariant curves. Namely, in Propositions \ref{p:p1} and \ref{p:p2} we use the algorithm above to proof non-Liouvillian integrability of \eqref{eq:Ince_XV} and \eqref{eq:Ince_V} and classify their transcendental invariant curves. In Section 5 we provide an example to illustrate the application of the algorithm for the proof of non-Puiseux integrability.

\section{Integrability properties of \eqref{eq:Ince_XV}.}

The integrability properties of \eqref{eq:Ince_XV} can bs summarized in the following statement:

\begin{proposition}
\label{p:p1}
Equation \eqref{eq:Ince_XV} is not Liouvillian integrable, it possess a non-Liouvillian first integral
\begin{equation}\label{eq:Ince_XV_fi}
  I=H_{1}H_{2}^{-1},
\end{equation}
and an integrating factor
\begin{equation}\label{eq:Ince_XV_M}
 M=(H_{1}H_{2}H_{3})^{-1}.
\end{equation}
Here $H_{i}$, $i=1,2,3$ are invariant curves of \eqref{eq:Ince_XV} with their respective cofactors that are given by
\begin{equation}\label{eq:Ince_XV_IC}
\begin{gathered}
H_{1}=\mbox{Ai}^{\,'}\left\{\alpha w+\frac{(u-1)^{2}}{4\alpha^{2} w^{2}}\right\}+\frac{u-1}{2\alpha w} \mbox{Ai}\left\{\alpha w+\frac{(u-1)^{2}}{4\alpha^{2} w^{2}}\right\}, \quad \lambda_{1}=\frac{u-1}{2w},\\
H_{2}=\mbox{Bi}^{\,'}\left\{\alpha w+\frac{(u-1)^{2}}{4\alpha^{2} w^{2}}\right\}+\frac{u-1}{2\alpha w} \mbox{Bi}\left\{\alpha w+\frac{(u-1)^{2}}{4\alpha^{2} w^{2}}\right\}, \quad \lambda_{2}=\frac{u-1}{2w},\\
H_{3}=w, \quad \lambda_{3}=\frac{u}{w},
 \end{gathered}
\end{equation}
where $u=w_{\tau}$, $\mbox{Ai}$ and $\mbox{Bi}$ are the Airy functions and the prime denotes a derivative with respect to the argument of the corresponding function.

\end{proposition}
\begin{proof}
Equation \eqref{eq:Ince_XV} is equivalent to the following dynamical system
\begin{equation}\label{eq:p1_1}
  w_{\tau}=u, \quad u_{\tau}=u^{2}/w-u/w-2\alpha^{3}w^{2}.
\end{equation}
One can find the invariant curve $H_{3}$ of \eqref{eq:p1_1} directly from definition \eqref{eq:L_def} assuming that $H$ depends only on $w$.

Now we apply the algorithm presented about to study integrability of \eqref{eq:p1_1}. With the help of scaling transformations $\tau\,'=\alpha \tau$, $w'=\alpha w$ we set $\alpha$ equal to $1$ in \eqref{eq:p1_1} (the primes are omitted in what follows). Further we use the transformations
\begin{equation}\label{eq:p1_2}
w=v-z^{2}, \quad u=2z(v-z^{2})+1,
\end{equation}
which can be inverted as follows
\begin{equation}\label{eq:p1_3}
z=\frac{u-1}{2w}, \quad v=\left(\frac{u-1}{2w}\right)^{2}+w.
\end{equation}
As a result, from \eqref{eq:p1_1}, taking into account that $\alpha=1$, we obtain
\begin{equation}\label{eq:p1_4}
v_{\tau}=1, \quad z_{\tau}=z^{2}-v.
\end{equation}
Consequently, integrability of \eqref{eq:p1_4} and \eqref{eq:p1_1} is evident now, since \eqref{eq:p1_4} can be converted into a Riccati equation.

Now we demonstrate that system \eqref{eq:p1_4} does not possess algebraic invariant curves and it is not Liouvillian integrable. System \eqref{eq:p1_4} is equivalent to
\begin{equation}\label{eq:p1_5}
z_{v}=z^{2}-v.
\end{equation}
The leading order terms of \eqref{eq:p1_5} corresponding to the asymptotic at $v=\infty$ are the terms on the right-hand side of \eqref{eq:p1_5}, which contain no derivatives. Thus, we see that there is no resonance corresponding to this asymptotic and its coefficients are uniquely determined \cite{Bruno2000,Bruno2004}. As a result, it is easy to demonstrate that \eqref{eq:p1_4} admits two Puiseux series around $v=\infty$ without arbitrary parameters, which are
\begin{equation}\label{eq:p1_5a}
z=\pm \sqrt{v}+\frac{1}{4v}\mp\frac{5}{32v^{5/2}}+\ldots
\end{equation}
Thus, the maximal degree of an irreducible algebraic invariant curve of \eqref{eq:p1_4} with respect to $z$ is $2$ (see, \cite{Demina2022}). Since the degree of \eqref{eq:p1_4} is also 2, the degree of the cofactor of an invariant curve is at most $1$ with respect to $z$. We assume that the cofactor is polynomial with respect to both $z$ and $v$ and look for invariant curves of \eqref{eq:p1_4} in the form
\begin{equation}\label{eq:p1_6}
H=A(v)z^{2}+B(v)z+C(v), \quad \lambda=e_{1}z+e_{2}v+e_{3}.
\end{equation}
After substituting \eqref{eq:p1_6} into the definition of an invariant curve for \eqref{eq:p1_4} and equating coefficients at different powers of $z$ we obtain
\begin{equation}\label{eq:p1_7}
\begin{gathered}
(e_{1}-2)A=0, \quad A_{v}-(e_{2}v+e_{3})A-(e_{1}-1)B=0, \\ B_{v}-(e_{2}v+e_{3})B-2vA-e_{1}C=0, \quad
C_{v}-(e_{2}v+e_{3})C-vB=0.
\end{gathered}
\end{equation}
Applying transformations $A=\exp\{v(e_{2}v+2e_{3})/2\}\tilde{A}$, $B=\exp\{v(e_{2}v+2e_{3})/2\}\tilde{B}$ and $C=\exp\{v(e_{2}v+2e_{3})/2\}\tilde{C}$ from \eqref{eq:p1_7} we obtain
\begin{equation}\label{eq:p1_7a}
\begin{gathered}
(e_{1}-2)\tilde{A}=0, \quad \tilde{A}_{v}-(e_{1}-1)\tilde{B}=0, \quad \tilde{B}_{v}-2v\tilde{A}-e_{1}\tilde{C}=0, \quad \tilde{C}_{v}-v\tilde{B}=0.
\end{gathered}
\end{equation}
Suppose that $\tilde{A}\neq0$ (i.e.  $A\neq0$). Then we set $e_{1}=2$ and from \eqref{eq:p1_7a} we find that
\begin{equation}\label{eq:p1_8}
\tilde{B}=\tilde{A}_{v}, \quad \tilde{C}=\frac{1}{2}\tilde{A}_{vv}-v\tilde{A},
\end{equation}
where $\tilde{A}$ satisfies the equation
\begin{equation}\label{eq:p1_9}
\tilde{A}_{vvv}-4v\tilde{A}_{v}-2\tilde{A}=0.
\end{equation}
This equation is the second symmetric power of the Airy equation $s_{vv}-vs=0$ (see, \cite{Whittaker,Chudnovsky,Bronstein,Almkvist}) and, therefore, its general solution has the form
\begin{equation}\label{eq:p1_10}
\tilde{A}=c_{1}\mbox{Ai}^{2}(v)+c_{2}\mbox{Ai}(v)\mbox{Bi}(v)+c_{3}\mbox{Bi}^{2}(v),
\end{equation}
$c_{i}$, $|c_{1}|^{2}+|c_{2}|^{2}+|c_{3}|^2\neq0$, $i=1,2,3$ are arbitrary constants. The Airy equation has no Liouvillian solutions (see, e.g. \cite{Kovacic1986}). Thus, system \eqref{eq:p1_4} does not have polynomial invariant curves of second degree with respect to $z$. Notice that we do not construct explicitly invariants of \eqref{eq:p1_4} and their cofactors in this case, because the integrability of \eqref{eq:p1_4} can be determined by invariant curves of degree 1 with respect to $z$.

Suppose that $\tilde{A}=0$ and $\tilde{B}\neq0$. Then from \eqref{eq:p1_7a} we have that $e_{1}=1$ and
\begin{equation}\label{eq:p1_11}
\tilde{B}_{vv}-v\tilde{B}=0, \quad \tilde{C}=\tilde{B}_{v}.
\end{equation}
We see that the function $B$ satisfies the equation for the Airy function (see, e.g. \cite{Olver}). Thus, there are no polynomial invariant curves of \eqref{eq:p1_4} of degree 1 with respect to $z$.

Let us explicitly construct transcendental invariants of \eqref{eq:p1_4} in the case of $\widetilde{A}=0$ and $\widetilde{B}\neq0$. We suppose that $e_{2}=e_{3}=0$ since it does not impact the integrability of \eqref{eq:p1_4}. As a result, we find two invariants
\begin{equation}\label{eq:p1_12}
\begin{gathered}
\widetilde{H}_{1}=\mbox{Ai}(v)z+\mbox{Ai}(v)^{'}, \quad \tilde{\lambda}_{1}=z,\\
\widetilde{H}_{2}=\mbox{Bi}(v)z+\mbox{Bi}(v)^{'}, \quad \tilde{\lambda}_{2}=z.
\end{gathered}
\end{equation}
Below we show that these invariants actually define integrability of \eqref{eq:p1_4} and \eqref{eq:p1_1}.

If we proceed further and assume that $\tilde{A}=\tilde{B}=0$ from \eqref{eq:p1_7a} we get that $e_{1}=0$ and $\tilde{C}$ is an arbitrary constant. This leads to the existence of the exponential invariant of \eqref{eq:p1_4}
\begin{equation}\label{eq:p1_11c}
E_{1}=\exp\{v(e_{2}v+2e_{3})/2\}, \quad L_{1}=e_{2}v+e_{3}.
\end{equation}
Since the divergence of the vector field associated to \eqref{eq:p1_4} depends only on $z$, this exponential invariant has no impact on integrability of \eqref{eq:p1_4}. Taking into account that there are no algebraic invariant curves of \eqref{eq:p1_4}, system \eqref{eq:p1_4} can be Liouvillian integrable if and only if it has a Darboux integrating factor of the form $R= \exp\{b(z,v)\}$, where $b(z,v)\in\mathbb{C}[z,v]$. The function $b$ satisfies the equation
\begin{equation}\label{eq:p1_11a}
(z-v^2)b_{z}+b_{v}+2z=0,
\end{equation}
whose general solution is
\begin{equation}\label{eq:p1_11b}
b=2\ln\left\{\mbox{Bi}(v)-\frac{\widetilde{H}_{2}}{\widetilde{H}_{1}}\right\}+S\left(\frac{\widetilde{H}_{2}}{\widetilde{H}_{1}}\right),
\end{equation}
where $\widetilde{H}_{1}$, $\widetilde{H}_{1}$ are given in \eqref{eq:p1_12} and $S$ is an smooth arbitrary function of its argument. From \eqref{eq:p1_11b} it is clear that $b$ is not polynomial and we conclude that system \eqref{eq:p1_4} is not Liouvillian integrable. Since transformations \eqref{eq:p1_2} (or \eqref{eq:p1_3}) preserve Liouvillian integrability, we obtain that dynamical system \eqref{eq:p1_1} and equation \eqref{eq:Ince_XV} are not Liouvillian integrable as well.

Using \eqref{eq:p1_12} and Proposition \ref{p:p5} it is clear that system \eqref{eq:p1_4} has the first integral and the integrating factor
\begin{equation}\label{eq:p1_13}
\tilde{I}=\widetilde{H}_{1}/\widetilde{H}_{2}, \quad \widetilde{M}=\left(\widetilde{H}_{1}\widetilde{H}_{2}\right)^{-1}.
\end{equation}
Notice that neither $\widetilde{H}_{1}$ nor $\widetilde{H}_{2}$ is a Puiseux polynomial at $v=\infty$, since the Airy functions do not have Puiseux expansions at the infinity, because their asymptotics there contain exponential terms. However the product $\widetilde{H}_{1}\widetilde{H}_{2}$ is a Puiseux polynomial of second degree with respect to $z$ because in terms like $\mbox{Ai}(v)\,\mbox{Bi}(v)$, $\mbox{Ai}^{'}(v)\,\mbox{Bi}(v)$ and $\mbox{Ai}(v)\,\mbox{Bi}^{'}(v)$ the exponential factors are cancelled out. Therefore, dynamical system \eqref{eq:p1_4} is Puiseux integrable at infinity.  Clearly that  $\widetilde{H}_{1}$ and $\widetilde{H}_{2}$  are Puiseux polynomials at the origin and system \eqref{eq:p1_4} is Puiseux integrable around $v=0$.

With the help of the transformations \eqref{eq:p1_3} one can map invariant curves \eqref{eq:p1_12} of \eqref{eq:p1_4} into invariant curves \eqref{eq:Ince_XV_IC} of \eqref{eq:p1_1}. Then integrability of \eqref{eq:p1_1} imminently follows from Proposition \ref{p:p5}.

Finally, we remark that integrating factor \eqref{eq:Ince_XV_M} is not of Puiseux type since neither invariant curves $H_{1}$ and $H_{2}$ nor their product are Puiseux polynomials. Therefore, system \eqref{eq:p1_1} and equation \eqref{eq:Ince_XV} are not Puiseux integrable.

This completes the proof.
\end{proof}

\begin{remark}
From the proof of Proposition \ref{p:p1} it follows that the Puiseux integrability is not conserved by algebraic transformations. This is not the case for the Liouvillian integrable systems, since it is clear that algebraic transformations preserve Liouvillian functions and the existence of a Liuvillian first integral yields the existence of a Darboux integrating factor. Therefore, there is a question of choosing the most convenient coordinate frame to study a dynamical system. Often it is provided by the algebraic invariant curves of the corresponding dynamical system (see \cite{GS}).
\end{remark}

\section{Integrability properties of \eqref{eq:Ince_V}.}

In the following statement the integrability of \eqref{eq:Ince_V} is established.

\begin{proposition}
\label{p:p2}
For an arbitrary value of the parameter $\beta$ equation \eqref{eq:Ince_V} is not Liouvillian integrable, it has the first integral
\begin{equation}\label{eq:Ince_V_FI}
I=H_{5}H_{6}^{-1},
\end{equation}
and the integrating factor
\begin{equation}\label{eq:Ince_V_M}
M=\left(H_{4}H_{5}H_{6}\right)^{-1},
\end{equation}
where $H_{4}$, $H_{5}$ and $H_{6}$ are invariant curves of \eqref{eq:Ince_V} with their respective cofactors
\begin{equation}\label{eq:Ince_V_IC}
\begin{gathered}
H_{4}=u+w^{2}-\beta^{2}, \quad \lambda_{4}=1, \quad u=w_{\tau}, \\
H_{5}=(w-\beta)I_{2\beta}\{2\sqrt{H_{4}}\}-\sqrt{H_{4}} I_{2\beta+1}\{2\sqrt{H_{4}}\}, \quad \lambda_{5}=-w,\\
H_{6}=(w-\beta)K_{2\beta}\{2\sqrt{H_{4}}\}+\sqrt{H_{4}} K_{2\beta+1}\{2\sqrt{H_{4}}\}, \quad \lambda_{6}=-w.
\end{gathered}
\end{equation}
Here $I_{2\beta}\{\theta\}$ and $K_{2\beta}\{\theta\}$ are the modified Bessel functions of the first and second kind, respectively.

Equation \eqref{eq:Ince_V} is Liouvillian integrable if and only if first integral \eqref{eq:Ince_V_FI} degenerates into an elementary function, i.e. at
\begin{equation}\label{eq:Ince_V_LIC}
2\beta=\pm n \pm \frac{1}{2},\quad n\in\mathbb{N}\cup\{0\}.
\end{equation}
\end{proposition}
\begin{proof}
To proof this statement we convert \eqref{eq:Ince_V} into the following equivalent dynamical system
\begin{equation}\label{eq:p2_1}
 w_{\tau}=u, \quad u_{\tau}=-2wu+u+w^{2}-\beta^{2}.
\end{equation}
It is straightforward to demonstrate that \eqref{eq:p2_1} admits invariant curve $H_{4}$. Applying the above algorithm we use the transformation
\begin{equation}\label{eq:a1_1}
  u=s^{2}-w^{2}+\beta^{2}.
\end{equation}
that simplifies system \eqref{eq:p2_1}. As a result, from \eqref{eq:p2_1} we get
\begin{equation}\label{eq:a1_2}
  w_{\tau}=s^{2}-w^{2}+\beta^{2}, \quad s_{\tau}=s/2.
\end{equation}
Now it is clear that systems \eqref{eq:a1_2} and \eqref{eq:p2_1} are integrable since \eqref{eq:a1_2} is equivalent to the Riccati equation.

Let us prove that \eqref{eq:a1_2} and, hence \eqref{eq:p2_1} and \eqref{eq:Ince_V} are Liouvillian integrable if and only if \eqref{eq:Ince_V_LIC} holds. First, observe that transformation \eqref{eq:a1_1} preserves Liouvillian integrability. Second, for arbitrary $\beta$ system \eqref{eq:a1_2} admits one independent of $w$ invariant algebraic curve that is
\begin{equation}\label{eq:a1_3a}
\widetilde{H}_{3}=s, \quad \tilde{\lambda}_{3}=1/2.
\end{equation}
Third, we demonstrate that for arbitrary $\beta$ system \eqref{eq:a1_2} does not have algebraic invariant curves apart from $\widetilde{H}_{3}$. Indeed, consider the first-order differential equation associated to \eqref{eq:a1_2}
\begin{equation}\label{eq:a1_3}
  sw_{s}=2(s^{2}-w^{2}+\beta^{2}).
\end{equation}
Since the first two terms on the right-hand side of \eqref{eq:a1_3} give the highest degree upon substitution $w\sim s^{p}$ they are the leading order terms for asymptotic at $s=\infty$. These terms do not contain derivative and, hence, the corresponding asymptotic at $s=\infty$ does not contain arbitrary coefficients \cite{Bruno2000,Bruno2004}. Therefore, it is not difficult to demonstrate that solutions of \eqref{eq:a1_3} has only two Laurent expansions near $s=\infty$, which are
\begin{equation}\label{eq:a1_4}
 w=\pm s -\frac{1}{4}\pm \frac{16\beta^{2}-1}{32s}+\ldots\ ,
\end{equation}
Thus, we conclude that the maximal degree of an irreducible invariant curve of \eqref{eq:a1_2} with respect to $w$ is 2 (see, e.g. \cite{Demina2022}). Since the degree of \eqref{eq:a1_2} is also 2, the maximal degree of a polynomial with respect to $w$ cofactor is 1 (see \cite{Zhang}).

As in the case of Proposition \ref{p:p1}, integrability of \eqref{eq:a1_2} can be established only with the help of invariant curves of first degree with respect to $w$. In additional, it is straightforward to verify that for arbitrary $\beta$ invariant curves of second degree with respect to $w$ are expressed in terms of the modified Bessel functions and are algebraic if and only if \eqref{eq:Ince_V_LIC} is satisfied.

Therefore, we proceed with invariant curves of degree 1 with respect to $w$ and substitute
\begin{equation}\label{eq:a1_5}
 H=B(s)w+C(s), \quad \lambda=e_{1}w+e_{2}s+e_{3},
\end{equation}
into the definition of an invariant curve for \eqref{eq:a1_2}. Here $B(s)\not\equiv0$.

As a consequence, collecting coefficients at different powers of $w$ we obtain
\begin{equation}\label{eq:a1_6}
\begin{gathered}
 e_{1}+1=0, \quad  sB_{s}-2(e_{2}s+e_{3})B-2e_{1}C=0 \\
sC_{s}-2(e_{2}s+e_{3})C+2(s^{2}+\beta^{2})B=0.
 \end{gathered}
\end{equation}
Eliminating $C$ between the last two equations of \eqref{eq:a1_6} we obtain
\begin{equation}\label{eq:a1_7}
\begin{gathered}
s^{2}B_{ss}-(4[e_{2}s+e_{3}]-1)sB_{s}+\left(4[e_{2}s+e_{3}]^2-4[s^{2}+\beta^{2}]-2e_{2}s\right)B=0.
 \end{gathered}
\end{equation}
With the help of the transformation $B=s^{2e_{3}}{\rm e}^{2e_{2}s}B'$ from \eqref{eq:a1_7} we obtain that $B'$ satisfies the equation
\begin{equation}\label{eq:a1_8}
s^{2}B_{ss}^{'}+sB_{s}^{'}-4(s^{2}+\beta^{s})B^{'}=0,
\end{equation}
which is the equation for the modified Bessel function up to the scaling $\beta\rightarrow\beta/2$, $s\rightarrow s/2$. Consequently, we find that the general solution of \eqref{eq:a1_7} is
\begin{equation}\label{eq:a1_9}
B=s^{2e_{3}}{\rm e}^{2e_{2}s}\left(c_{4}I_{2\beta}(2s)+c_{5}K_{2\beta}(2s)\right).
\end{equation}
Here $c_{4}$ and $c_{5}$ are arbitrary constants. It is known that the Bessel functions degenerate into elementary ones if and only if their order is half integer \cite{Olver,Bateman,Kovacic1986}. Therefore, invariant curve \eqref{eq:a1_5} can be algebraic if and only if $2\beta=\pm n \pm 1/2$. In what follows we use the notation
\begin{equation}\label{eq:a1_10}
a_{k}=\frac{(n+k)!}{k!(n-k)!4^{k}}\,.
\end{equation}
Then we have that
\begin{equation}\label{eq:a1_11}
I_{\pm n \pm 1/2}(2s)=\frac{1}{\sqrt{2\pi}}\left({\rm e}^{2s}\sum\limits_{k=0}^{n}\frac{(-1)^{k}a_{k}}{s^{k+1/2}}\mp (-1)^{n}{\rm e}^{-2s}\sum\limits_{k=0}^{n}\frac{a_{k}}{s^{k+1/2}}\right),
\end{equation}
and
\begin{equation}\label{eq:a1_12}
K_{\pm n \pm 1/2}(2s)=\frac{\sqrt{\pi}}{2} {\rm e}^{-2s}\sum\limits_{k=0}^{n}\frac{a_{k}}{s^{k+1/2}},
\end{equation}
Now it is clear that there are two possibilities when an invariant curve \eqref{eq:a1_5} is algebraic and depends on $w$. In the first case we set $c_{4}=0$, $e_{1}=-e_{2}=-1$, $e_{3}=n/2+1/4$ and $c_{5}=2/\sqrt{\pi}$. Substituting these values into \eqref{eq:a1_12}, \eqref{eq:a1_9}, \eqref{eq:a1_7} and \eqref{eq:a1_5} we obtain that there is a polynomial invariant curve of \eqref{eq:a1_2} with the cofactor
\begin{equation}\label{eq:a1_13}
\tilde{\lambda}_{4}=-w+s+\frac{n}{2}+\frac{1}{4}
\end{equation}
In the same way one can demonstrate that at  $c_{5}=((-1)^{n}c_{4})/\pi$, $e_{1}=e_{2}=-1$, $e_{3}=n/2+1/4$ and $c_{4}=2\sqrt{\pi}$ there is another polynomial invariant curve of \eqref{eq:a1_1} with the cofactor
\begin{equation}\label{eq:a1_14}
\tilde{\lambda}_{5}=-w-s+\frac{n}{2}+\frac{1}{4}.
\end{equation}
It is easy to see that
\begin{equation}\label{eq:a1_14f}
-\tilde{\lambda}_{4}-\tilde{\lambda}_{5}+2n\tilde{\lambda}_{3}=-\mbox{div}\mathcal{X},
\end{equation}
where $\mathcal{X}$ is the vector field associated to \eqref{eq:a1_2}.

Therefore, \eqref{eq:a1_2} has a Darboux integrating factor and a Liouvillian first integral if $2\beta$ is half-integer. The latter can be easily constructed in the explicit form taking into account that \eqref{eq:a1_2} admits an exponential factor $E_{2}={\rm e}^{s}$ with the cofactor $L_{2}=s/2$. Indeed, one can see that $\tilde{\lambda}_{4}-\tilde{\lambda}_{5}-4L_2=0$ and, hence, at $2\beta=\pm n\pm 1/2$ system \eqref{eq:a1_2} has a rational first integral with an exponential factor that has the form
\begin{equation}\label{eq:a1_14a}
I=\frac{{\rm e}^{-4s}\left[(4w+4s+1)\sum\limits_{k=0}^{n}a_{k}s^{n-k}+2\sum\limits_{k=0}^{n}ka_{k}s^{n-k}\right]}{(4w-4s+1)\sum\limits_{k=0}^{n}(-1)^{k}a_{k}s^{n-k}+2\sum\limits_{k=0}^{n}(-1)^{k}ka_{k}s^{n-k}}.
\end{equation}

Now we study integrability of \eqref{eq:a1_2} at $2\beta\neq\pm n\pm 1/2$. We see that in this case the only algebraic invariant curve of \eqref{eq:a1_2} is \eqref{eq:a1_3}. Let us prove that the unique exponential factor of \eqref{eq:a1_2} at arbitrary $\beta$ is $E_{1}={\rm e}^{s}$.  The most general form of the exponential factor is
\begin{equation}\label{eq:a1_14b}
E_{3}=\exp\left\{\frac{b(w,s)}{s^{l}}\right\}, \quad l\in\mathbb{N}\cup\{0\},
\end{equation}
where $b(w,s)$ is a polynomial in $w$ and $s$, which is coprime with $s$. It is clear that the maximal degree of the cofactor of $E_{3}$ is one and substituting \eqref{eq:a1_14b} into \eqref{eq:L_def} with $L_{3}=e_{1}w+e_{2}s+e_{3}$ we obtain
\begin{equation}\label{eq:a1_14c}
2(\beta^{2}+s^{2}-w^{2})b_{w}+sb_{s}-lb-2(e_{1}w+e_{2}s+e_{3})s^{l}=0
\end{equation}
Assume that $l\neq0$. Then $b(w,s)|_{s=0}=\tilde{b}(w)\neq0$ (because $b$ and $s$ are coprime) satisfies the equation
\begin{equation}\label{eq:a1_14d}
2(\beta^{2}-w^{2})b_{w}-lb=0,
\end{equation}
whose general solution is
\begin{equation}\label{eq:a1_14e}
b=c_{6}(w-\beta)^{-\frac{l}{2\beta}}(w+\beta)^{\frac{l}{2\beta}},
\end{equation}
where $c_{6}\neq0$ is an arbitrary constant. This function is not polynomial unless $l$ is zero.

Suppose that $l=0$. Then directly integrating \eqref{eq:a1_14c} we find that for arbitrary $\beta$ the only polynomial solution is $b=2e_{2}s$ at $e_{3}=e_{1}=0$. Consequently, the only exponential factor of \eqref{eq:p3_2} for arbitrary $\beta$ is $E_{2}$ with the cofactor $L_{2}$. Since the divergence of \eqref{eq:a1_2} depends on $w$ and $L_{2}$ and $\tilde{\lambda}_{3}$ do not, we see that for arbitrary $\beta$ system \eqref{eq:a1_2} does not have a Darboux integrating factor. Thus, system \eqref{eq:a1_2} is Liouvillian integrable if and only if \eqref{eq:Ince_V_LIC} holds.

Suppose that $e_{2}=e_{3}=0$ in \eqref{eq:a1_9}, which has no impact on integrability of \eqref{eq:a1_2}. Then, setting either $c_{4}$ or $c_{5}$ to zero in \eqref{eq:a1_9} we obtain two invariant curves of \eqref{eq:a1_2}
\begin{equation}\label{eq:a1_15}
\begin{gathered}
\widetilde{H}_{6}=(w-\beta)I_{2\beta}(2s)-sI_{2\beta+1}(2s), \quad \tilde{\lambda}_{6}=-w,\\
\widetilde{H}_{7}=(w-\beta)K_{2\beta}(2s)+sK_{2\beta+1}(2s), \quad \tilde{\lambda}_{7}=-w
\end{gathered}
\end{equation}
Note that $a\tilde{\lambda}_{6}-(2+a)\tilde{\lambda}_{7}-\tilde{\lambda}_{3}=-\mbox{div}\mathcal{X}$, where $\mathcal{X}$ is the vector field associated to \eqref{eq:a1_2}, holds for any $a\in\mathbb{C}$. Then, we see from Proposition \ref{p:p5} that \eqref{eq:a1_2} is integrable. The parameter $a$ can be chosen arbitrarily, taking into account that an integrating factor is defined by a multiplication on a first integral. Assuming that $a=-1$  we get that
\begin{equation}\label{eq:a1_16}
\widetilde{M}=\left(\widetilde{H}_{6}\widetilde{H}_{7}\widetilde{H}_{3}\right)^{-1},
\end{equation}
is an integrating factor of \eqref{eq:a1_2}. Since the modified Bessel functions do not have Puiseux asymptotics at the infinity, neither $\widetilde{H}_{6}$ nor $\widetilde{H}_{7}$ is a Puiseux polynomial. However their product is a second degree Puiseux polynomial with respect to $w$ because exponential terms are cancelled out in the asymptotics for the products of the modified Bessel functions. Thus, we see that system \eqref{eq:a1_2} is Puiseux integrable. The first integral of \eqref{eq:a1_2} has the from $\widetilde{H}_{6}/\widetilde{H}_{7}$ since $\tilde{\lambda}_{6}-\tilde{\lambda}_{7}=0$. Let us remark that $\widetilde{H}_{6}$ is a Puiseux polynomial around the origin for arbitrary $\beta$, while this is not the case for $\widetilde{H}_{7}$, since for some values of the index the asymptotic around the origin for the modified Bessel function of the second kind contains logarithmic terms.

Inverting transformation \eqref{eq:a1_1} we obtain expressions \eqref{eq:Ince_V_FI}, \eqref{eq:Ince_V_IC} and \eqref{eq:Ince_V_M}. Let us remark that one can choose any sign in front of the square root while inverting \eqref{eq:a1_1}. We use the branch with the plus sign. Integrating factor \eqref{eq:Ince_V_M} is not of Puisex type, since in \eqref{eq:Ince_V_M} both variables are inside the argument of the Bessel function and, hence, \eqref{eq:Ince_V_M} cannot be a polynomial in one of the variables. Consequently, system \eqref{eq:p2_1} and equation \eqref{eq:Ince_V} are not Puiseux integrable for $2\beta\neq\pm n\pm 1/2$.

Finally, let us note that system \eqref{eq:p1_1} is system (1.2) from \cite{Demina2020} with $\sigma_{0}=-1$ and $\delta_{0}=\beta^{2}-1/4$. Therefore, in order to prove that \eqref{eq:p1_1} is Liouvillian integrable if and only if \eqref{eq:Ince_V_LIC} is satisfied one can use Theorem 1.2a from \cite{Demina2020}. In addition, is should be possible to transform first integral (1.6) from \cite{Demina2020} into \eqref{eq:a1_14a} and vice verse.  On the other hand, the integrability of \eqref{eq:p1_1} for arbitrary value of the parameter $\beta$ was not considered in \cite{Demina2020}.
This completes the proof.
\end{proof}

\section{Example of non-Puiseux integrable oscillator.}

In this section we study integrability of \eqref{eq:ie_1a}. It is equivalent to the following dynamical system
\begin{align}\label{eq:ie_1}
  w_{t}&= w(w-1)(w^{2}+u)+(b-1)w-a, \nonumber \\
  u_{t}&=-(2w^{3}-2w^{2}-1)(w^{2}+u)-2w(bw-w-a).
\end{align}
We assume that $a,b\not\in\mathbb{Z}$ and $a-b\neq-(1+q)$, $q\in\mathbb{N}\cup\{0\}$. Under these assumptions it is possible to express invariant curves of \eqref{eq:ie_1} in terms of the Kummer function. Other values of the parameters $a$ and $b$ should be considered elsewhere.

Let us apply the algorithm presented above to demonstrate integrability of \eqref{eq:ie_1}. We begin with the first step. It is easy to verify that \eqref{eq:ie_1} possesses an algebraic invariant curve $H=u+w^{2}$ with the cofactor $\lambda=1$. Consequently, we make the change of variables
\begin{equation}\label{eq:ie_2}
  z=w, \quad y=u+w^{2},
\end{equation}
which can be inverted as follows
\begin{equation}\label{eq:ie_3}
  w=z, \quad u=y-z^{2}.
\end{equation}
As a result, we obtain
\begin{gather}\label{eq:ie_4}
 y_{t}=y, \quad z_{t}= yz^{2}+(b-1-y)z-a.
\end{gather}
This dynamical system is equivalent to the Riccati equation
\begin{equation}\label{eq:ie_5}
 y z_{y}=yz^{2}+(b-1-y)z-a,
\end{equation}
and, hence, it is clear that both \eqref{eq:ie_1} and \eqref{eq:ie_4} are integrable.

Now we proceed with the algorithm in order to find transcendental invariant curves of \eqref{eq:ie_5} and \eqref{eq:ie_1} and establish different types of the integrability of \eqref{eq:ie_5} and \eqref{eq:ie_1}, depending on the values of the parameters $a$ and $b$.

At the second step, we privilege variable $z$ and assume that transcendental invariant curves of \eqref{eq:ie_5} are polynomials with respect to $z$. One can find the maximal degree of an invariant curve with respect to $z$ as follows. There are two distinct asymptotic for the solution of \eqref{eq:ie_5} at $y=\infty$. The corresponding dominant balances are $z^{2}-zy$ and $yz+a$, which do not contain terms with derivatives and, hence, do not produce any resonances. Therefore, we find the following asymptotic expansions at $y=\infty$ with uniquely determined coefficients
\begin{equation}\label{eq:ie_6a}
 Z^{(1)}(y)=y-b+2+\frac{a+b-2}{y}+\frac{(b-4)(a+b-2)}{y^{2}}+\ldots,
\end{equation}
\begin{equation}\label{eq:ie_6b}
 Z^{(2)}(y)=-\frac{a}{y}-\frac{ab}{y^{2}}+\frac{a(a-b-b^{2})}{y^{3}}+\ldots.
\end{equation}
Thus, the maximal degree of an invariant curve of \eqref{eq:ie_4} with respect to $z$ is 2 and we suppose that transcendental invariant curves  of \eqref{eq:ie_4} have the form
\begin{equation}\label{eq:ie_7}
  H=A(y)z^{2}+B(y)z+C(y).
\end{equation}
Here functions $A$, $B$ and $C$ are supposed to be sufficiently smooth and $|A|^{2}+|B|^{2}+|C|^{2}\neq0$.

Since the vector field \eqref{eq:ie_4} is polynomial of degree 3, we will assume that the cofactor has the form
\begin{equation}\label{eq:ie_8}
  \lambda=e_{6}z^2+e_{5}yz+e_{4}y^2+e_{3}z+e_{2}y+e_{1},
\end{equation}
where $e_{j}\in\mathbb{C}$, $j=\overline{1,6}$ are arbitrary parameters. This completes \textit{the third step of the algorithm}.

At the \textit{fourth step} we substitute \eqref{eq:ie_7} and \eqref{eq:ie_8} into the definition of an invariant curve \eqref{eq:L_def} for \eqref{eq:ie_4}. As a result, we obtain a system of \textit{linear} algebraic and differential equations for $A$, $B$, $C$ and $e_{j}$, $j=\overline{1,6}$:
\begin{equation}\label{eq:ie_8a}
\begin{gathered}
e_{6}A=0, \quad ([e_{5}-2]y+e_{3})A+e_{6}B=0, \\
yA_{y}+(2b-e_{1}-2-(e_{2}+2)y-e_{4}y^{2})A-([e5-1]y+e_{3})B-e_{6}C=0,\\
yB_{y}+(b-e_{1}-1-(e_{2}+1)y-e_{4}y^{2})B-(e_{5}y+e_{3})C-2aA=0,\\
yC_{y}-(e_{4}y^{2}+e_{2}y+e_{1})C-aB=0.
  \end{gathered}
\end{equation}
Let us begin with the simplest case of invariant curves that depend only on $y$, i.e. we assume that $A=B=0$. As a consequence, we get that $e_{6}=e_{5}=e_{3}=0$ and
\begin{equation}\label{eq:ie_9}
C_{1}=c_{1}y^{e_{1}}\exp\left\{\frac{e_{4}y^{2}}{2}+e_{2}y\right\},
\end{equation}
where $c_{1}\neq0$ is an arbitrary constant. This essentially leads to an invariant curve and an exponential factor that depend only on $y$:
\begin{equation}\label{eq:ie_10}
\widetilde{H}_{8}=y, \quad \widetilde{\lambda}_{8}=1,  \quad E_{4}={\rm e}^{y}, \quad L_{4}=y,
\end{equation}
while other invariants that depend only on $y$ can be found from $\widetilde{H}_{8}$ and $E_{4}$.

Since the divergence of the vector field associated to \eqref{eq:ie_4} explicitly depends on $z$, one cannot establish its integrability using only \eqref{eq:ie_10}.

In what follows, without loss of generality, we assume that $e_{4}=e_{1}=e_{2}=0$. This can be done because the divergence of the vector field associated to \eqref{eq:ie_4} contains no terms proportional to $y^{2}$. The constant term and the term proportional to $y$ are covered by $\widetilde{H}_{8}$ and $E_{4}$. This is equivalent to the change of variables $A\rightarrow C_{1}\tilde{A}$, $B\rightarrow C_{1}\tilde{B}$, $C\rightarrow C_{1}\tilde{C}$ in system \eqref{eq:ie_8a}.

Suppose that $A=0$ and $B\neq0$. Then, from \eqref{eq:ie_8a} we obtain that $e_{6}=e_{3}=0$, $e_{5}=1$ and $B$ and $C$ satisfy the following system of equations
\begin{equation}\label{eq:ie_11}
yB_{y}+(b-1-y)B-yC=0, \quad yC_{y}-aB=0.
\end{equation}
Eliminating $B$ from \eqref{eq:ie_11} we find that $C$ is a solution the equation
\begin{equation}\label{eq:ie_12}
yC_{yy}+(b-y)C_{y}-aC=0.
\end{equation}
Provided that the above assumptions on the parameters $a$ and $b$ hold, the general solution of this equation is expressed in terms of the Kummer function $\mathcal{M}(a,b,y)$ \cite{Olver,METF2022}
\begin{equation}\label{eq:ie_13}
C_{2}=c_{6}\mathcal{M}(a,b,y)+c_{7}y^{1-b}\mathcal{M}(1+a-b,2-b,y),
\end{equation}
where $c_{6}$ and $c_{7}$ are arbitrary constants such that $|c_{6}|^{2}+|c_{7}|^{2}\neq0$.

As a result, we find two transcendental invariant curves of \eqref{eq:ie_4}
\begin{equation}\label{eq:ie_14}
\begin{gathered}
 \widetilde{H}_{9}=(z-1)\mathcal{M}(a,b,y)-z\mathcal{M}(a+1,b,y), \quad \widetilde{\lambda}_{9}=zy,  \quad\quad\quad \quad\quad\quad\quad \\
\widetilde{H}_{10}=y^{1-b}\left\{ ([y+a-b]z+a)\mathcal{M}(a-b+1,2-b,y)-\right.  \quad\quad\quad\quad\quad \quad\quad\quad \\
 \quad\quad\quad\quad\quad \quad\quad\quad \quad\quad\quad \quad\quad\quad \left. (a-1)z\mathcal{M}(a-b,2-b,y) \right\}, \quad \widetilde{\lambda}_{10}=zy,
\end{gathered}
\end{equation}
Note that the invariant curve $\widetilde{H}_{10}$ has a multiplier $y^{1-b}$, which is $\widetilde{H}_{8}^{1-b}$. Therefore, instead of $\widetilde{H}_{10}$ one can consider $\widetilde{H}_{11}=\widetilde{H}_{10}y^{b-1}$ with the cofactor $\widetilde{\lambda}_{11}=zy+b-1$.

Consequently, we find the first integral and integrating factor of \eqref{eq:ie_4}
\begin{equation}\label{eq:ie_14a}
\begin{gathered}
I=\widetilde{H}_{9}\widetilde{H}_{10}^{-1}, \quad M=\widetilde{H}_{9}^{\mu}\widetilde{H}_{10}^{-2-\mu}\widetilde{H}_{8}^{-b}E_{4},
  \end{gathered}
\end{equation}
where $\mu\in\mathbb{C}$ is an arbitrary parameter. Let us remark that the parameter $\mu$ appears in \eqref{eq:ie_14a} because any integrating factor is defined up to a multiplication by a first integral. This completes the application of the algorithm and establishes the integrability of \eqref{eq:ie_4} and, thus, of \eqref{eq:ie_1} and \eqref{eq:ie_1a}.

For any value of the parameter $\mu$ and for arbitrary $a$ and $b$, satisfying conditions given after formula \eqref{eq:ie_1}, integrating factor given in \eqref{eq:ie_14a} is not a Puiseux polynomial around $\infty$ and, hence, globally. It is clear that $\widetilde{H}_{8}$ and the exponent of $E_{4}$ are Puiseux polynomial everywhere on $\mathbb{C}$. Thus, we have to consider $\widetilde{H}_{9}$ and $\widetilde{H}_{10}$. It is well known (see e.g. \cite{Olver}) that the Kummer function has essential singularity at infinity and its asymptotic at $\infty$ contains an exponential multiplier. Substituting the first term of the expansion for the Kummer function at $\infty$  into \eqref{eq:ie_14a} one can verify that it contains exponential multiplicator at any value of $\mu$ and, hence, is not of Puiseux type.

Furthermore, there are no other forms of the first integral and integrating factors of \eqref{eq:ie_4} than those given in \eqref{eq:ie_14a}. Indeed, an arbitrary smooth function of $I=\widetilde{H}_{9}\widetilde{H}_{10}^{-1}$ is the general solution of the partial differential equation for the first integral of \eqref{eq:ie_14a}, which can be easily obtained by solving \eqref{eq:ie_5}. The function $M=\widetilde{H}_{9}^{\mu}\widetilde{H}_{10}^{-2-\mu}\widetilde{H}_{8}^{-b}E_{4}$ is a solution of the partial differential equation for the integrating factor and can be constructed from the expression for the first integral. Therefore, the product $MR(I)$, where $R$ is an arbitrary smooth function, is the most general form of the integrating factor for \eqref{eq:ie_4}. All these expressions have no Puiseux expansions at $\infty$ for the values of $a$ and $b$ when the Kummer function does not degenerate. Therefore, for these values of $a$ and $b$ dynamical system \eqref{eq:ie_4} is not Puiseux integrable.

We see that the fact that the confluent hypergeometric equation \eqref{eq:ie_12} has an irregular singularity at infinity is an obstacle to the Puiseux asymptotic expansion at $\infty$ for the first integral and integrating factor of \eqref{eq:ie_4}. Moreover, there are no algebraic transformations that change the behaviour of confluent hypergeometric functions at $\infty$. Thus, we obtain the following result:

\begin{proposition}
\label{p:p_Puiseux}
Dynamical system \eqref{eq:ie_4} is not globally Puiseux integrable  up to an arbitrary algebraic change of the variables at the values of the parameters $a$ and $b$ for which the Kummer function does not degenerate.
\end{proposition}

Let us formally prove that the only algebraic invariant curve of \eqref{eq:ie_4} is $\widetilde{H}_{8}$ and system does not have Darboux integrating factor. To do this we find transcendental invariant curves of \eqref{eq:ie_4} that are of second degree with respect to $z$. Assuming again that $e_{4}=e_{1}=e_{2}=0$ and $A\neq0$ from system \eqref{eq:ie_8} we obtain
that $e_{5}=2$, $e_{3}=0$,
\begin{equation}\label{eq:ie_15}
\begin{gathered}
B=\frac{yC_{y}}{a}, \quad A=\frac{y}{2a^{2}}\left(yC_{yy}+(b-y)C_{y}-2aC\right),
  \end{gathered}
\end{equation}
and $C$ is a solution of the equation
\begin{equation}\label{eq:ie_16}
\begin{gathered}
y^2C_{yyy}+3y(b-y)C_{yy}+(2y^2-4(a+b)y+2b^2-b)C_{y}-4a\left(b-y-\frac{1}{2}\right)C=0.
  \end{gathered}
\end{equation}
The latter equation is the second symmetric power of \eqref{eq:ie_12} and, hence, its general solution is expressed in terms of $\mathcal{M}^{2}(a,b,y)$, $\widetilde{\mathcal{M}}^{2}(a,b,y)$ and $\mathcal{M}(a,b,y)\widetilde{\mathcal{M}}(a,b,y)$, where $\widetilde{\mathcal{M}}=y^{1-b}\mathcal{M}(1+a-b,2-b,y)$. Consequently, we obtain that for arbitrary $a$ and $b$ there are no algebraic invariant curves of \eqref{eq:ie_4} apart from $H_{18}$. Thus, there can exist a Darboux integrating factor if and only if there is an exponential factor of the form $E_{5}=\exp\left\{v(z,y)/y^{l}\right\}$, $l\in
\mathbb{N}\cup\{0\}$ and $v(z,y)\in\mathbb{C}[z,y]$. It is not difficult to verify that if $v_{z}=0$, the corresponding exponential factor is given by $E_{4}$. There are no polynomial solutions in the case of $v_{y}=0$. Thus, we assume that $v_{z}v_{y}\neq0$ and $v(z,y)=v_{m}(y)z^{m}+\ldots$, $m\in\mathbb{N}$ and $v_{m}\in \mathbb{C}[y]$. Upon substituting this expression for the exponential factor into \eqref{eq:L_def} and finding coefficient at the highest power of $z$ we get that $mv_{m}=0$, which is a contradiction. Thus, we see that \eqref{eq:ie_4} is not Liouvillian integrable for arbitrary $a$ and $b$. Notice that this result also follows from Proposition \ref{p:p_Puiseux}.

Finally, let us briefly consider the application of the test for Puiseux integrability of dynamical systems proposed in \cite{Demina2022} around $\infty$. The first step that consists in the classification of Puiseux series is already done (see \eqref{eq:ie_6a} and \eqref{eq:ie_6b}). The local forms of the cofactors corresponding to the series $Z^{(1,2)}(y)$  are
\begin{equation}\label{eq:ie_17}
  \lambda^{(1,2)}=y z+y Z^{(1,2)}(y)+b-1-y.
\end{equation}
Then, it is not difficult to verify that there is no exponential factors of the form $\exp\{u(y)/(z-Z^{(1,2)}(y))\}$. Consider exponential factors of the from $\exp\{v(y)z^{l}\}$ with the cofactor $L(y,z)$. The degree of $L(y,z)$ with respect to $z$ is $l+1$ and the divergence of the vector field associated to \eqref{eq:ie_4} is a linear function of $z$. Thus, $l=0$. Further, according to the algorithm of \cite{Demina2022} we suppose that there is an integrating factor
\begin{equation}\label{eq:ie_18}
M=y^{d_{0}}h(y){\rm e}^{v_{0}(y)}\prod\limits_{j=1}^{2}\left\{z-Z^{(j)}\right\}^{d_{j}}, \quad d_{j}\in\mathbb{C},
\end{equation}
and $v_{0}(y)$ and $h(y)$ are some Puiseux series around infinity. The cofactor of $y^{d_{0}}h(y){\rm e}^{v_{0}(y)}$ is
\begin{equation}\label{eq:ie_19}
\lambda_{0}=yV(y),
\end{equation}
where
\begin{equation}\label{eq:ie_20}
V(y)=\sum\limits_{k=0}^{\infty}h_{k}y^{1-k}.
\end{equation}
We substitute these expressions into the correlation
\begin{equation}\label{eq:ie_21}
d_{1}\lambda^{(1)}+d_{2}\lambda^{(2)}+\lambda_{0}+2yz+b-y=0,
\end{equation}
and collect coefficients at different powers of $z$. As a result, from the coefficient at $z$ we obtain that $d_{1}+d_{2}+2=0$ and from the rest of relations one can consequently find coefficients $h_{k}$, $k\in\mathbb{N}\cup\{0\}$.

Thus, we see that formally dynamical system \eqref{eq:ie_4} satisfies the algorithm of \cite{Demina2022} for Puiseux integrability. Moreover, this algorithm does not produce any relations on the parameters $a$ and $b$ that can lead to the Puiseux integrability. Thus, we see that the local algorithm of \cite{Demina2022} allows us to check only necessary conditions for the Puiseux integrability. The reason behind this is the fact that local expansions like \eqref{eq:ie_6a} and \eqref{eq:ie_6b} does not allow us to detect essential singularities. Actually, this is a well known phenomenon in the theory of the Painlev\'e equations (see, e.g. \cite{Ince,Ablowitz1980}), where it has been demonstrated that the existence of a local Laurent or Puiseux expansion does not guarantee the absence of an essential singularity.

On the other hand, if we use the algorithm proposed in this work, we can detect essential singularities using standard techniques for linear differential equations. Indeed, recall that the algorithm produces only linear correlations on the coefficients of an invariant curve, which is a polynomial with respect to the privileged variable. Consequently, these coefficients satisfy some linear differential equations with variable coefficients. What is more, for polynomial or rational vector fields the coefficients of these ordinary differential equations will be rational. The singularities of solutions of linear ordinary differential equations with rational coefficients, including singularities at $\infty$, can be classified by analysing only the coefficients without the need to find an explicit solution \cite{Forsyth,Novokshenov1986,Ilyashenko2008}. Thus, our algorithm allows one to deal with the cases where the integrating factor and/or the first integral have essential singularities. In other words, we can formally prove the integrability of a dynamical system and study singularities of invariant curves, a first integral and an integrating factor without explicitly constructing them. This can be summarized as follows:

\begin{remark}
  Suppose that we consider a polynomial or a rational dynamical system or a dynamical system that can be transformed into one of these. Then the linear ordinary differential equations for the coefficients of a polynomial with respect to the privileged variable invariant curve have rational coefficients. This allows one to classify singularities in the complex plane of invariant curves, a first integral and an integrating factor, if they exist, without explicitly constructing them. This can be done with the help of the general theory of linear differential equations (see, e.g.\cite{Forsyth,Novokshenov1986,Ilyashenko2008}).
\end{remark}

\section{Families of non-Liouvillian integrable oscillators}

Now let us build the equivalence class of \eqref{eq:Ince_XV} with respect to the \eqref{eq:GS}. In what follows we use an axillary function $l$  that is defined by
\begin{equation}\label{eq:l_def}
  l=hfg+fg_{x}-gf_{x}.
\end{equation}

The equivalence criterion for \eqref{eq:Lienard} and \eqref{eq:Ince_XV} can be summarized as follows
\begin{proposition}
\label{p:p3}
Suppose that $f\neq0$ and $g\neq0$. Then an equation from \eqref{eq:Lienard} is equivalent to \eqref{eq:Ince_XV} via \eqref{eq:GS} if and only if
\begin{equation}\label{eq:criterion_1}
6glf_{x}+3l^{2}-2fgl_{x}=0,
\end{equation}
where $l$ is given by \eqref{eq:l_def}. The functions that generate equivalence transformations \eqref{eq:GS} are
\begin{equation}\label{eq:p1_trans}
F^{3}=\frac{l}{4\alpha^{3}f^{3}}, \quad G=fF.
\end{equation}
It is assumed that $l\neq0$ and $fl_{y}-3lf_{y}\neq0$, since otherwise equivalence transformations degenerate.

Any equation from \eqref{eq:Lienard} that satisfies \eqref{eq:criterion_1} possesses a first integral
\begin{equation}\label{eq:integral_1}
  I=H_{7}H_{8}^{-1},
\end{equation}
and an integrating factor
\begin{equation}\label{eq:M_1}
M=\left(\frac{F_{x}}{G}\right)^{2} \left(H_{7}H_{8}H_{9}\right)^{-1},
\end{equation}
where
\begin{equation}\label{eq:Ic_1}
\begin{gathered}
H_{7}=\mbox{Ai}^{\,'}\left\{\alpha F+\frac{(F_{x}y-G)^{2}}{4\alpha^{2} F^{2}G^{2}}\right\}+\frac{F_{x}y-G}{2\alpha FG} \mbox{Ai}\left\{\alpha F+\frac{(F_{x}y-G)^{2}}{4\alpha^{2} F^{2}G^{2}}\right\},\\ \lambda_{7}=\frac{F_{x}y-G}{2F},\\
H_{8}=\mbox{Bi}^{\,'}\left\{\alpha F+\frac{(F_{x}y-G)^{2}}{4\alpha^{2} F^{2}G^{2}}\right\}+\frac{F_{x}y-G}{2\alpha FG} \mbox{Bi}\left\{\alpha F+\frac{(F_{x}y-G)^{2}}{4\alpha^{2} F^{2}G^{2}}\right\}, \\ \lambda_{8}=\frac{F_{x}y-G}{2F},\\
H_{9}=F, \quad \lambda_{9}=\frac{F_{x}y}{F}.
 \end{gathered}
\end{equation}
are invariant curves of equations from \eqref{eq:Lienard} defined by \eqref{eq:criterion_1}.
\end{proposition}
\begin{proof}
In order to obtain the necessary conditions, we substitute \eqref{eq:GS} into \eqref{eq:Ince_XV}. As a result, we arrive at \eqref{eq:Lienard} with
\begin{equation}\label{eq:p3_2}
h=\frac{F(F_{xx}-G_{x}F_{x})+GF_{x}^{2}}{GFF_{x}}, \quad f=\frac{G}{F}, \quad g=\frac{2\alpha^{3}F^{2}G^{2}}{F_{x}}.
\end{equation}
Therefore, equations which are equivalent to \eqref{eq:Ince_XV} via \eqref{eq:GS} are necessary of the from \eqref{eq:Lienard}. On the other hand, if for given values of $f$, $g$ and $h$ system \eqref{eq:p1_2} has a solution such that $GF_{x}\neq0$, then the corresponding equation from \eqref{eq:Lienard} is equivalent to \eqref{eq:Ince_XV}. Thus, if we find all $f$, $g$ and $h$, such that system \eqref{eq:p3_2} has a solution satisfying $GF_{x}\neq0$, we obtain the equivalence criterion between \eqref{eq:Ince_XV} and \eqref{eq:Lienard}.

This can be done as follows. First, we exclude $G$ from \eqref{eq:p1_2} using its second equation. As a consequence and after some simplifications, we get
\begin{equation}\label{eq:p3_3}
2gf F_{x}-Fl=0, \quad 4\alpha^{3}f^{3}F^{3}-l=0.
\end{equation}
Differentiating the second equation of \eqref{eq:p3_3} and substituting the result into \eqref{eq:p3_3} we find equivalence criterion \eqref{eq:criterion_1}. The expression for $F^{3}$ follows from the second equation of \eqref{eq:p3_3}. Further differentiation does not lead to any new compatibility conditions.

In order to find the expressions \eqref{eq:integral_1}, \eqref{eq:M_1} and \eqref{eq:Ic_1} we use Proposition 1 from \cite{Sinelshchikov2023a} and first integral \eqref{eq:Ince_XV_fi}, integrating factor \eqref{eq:Ince_XV_M} and invariant curves \eqref{eq:Ince_XV_IC} for \eqref{eq:Ince_XV}. This completes the proof.
\end{proof}

Now we present the equivalence class of \eqref{eq:Ince_V} with respect to \eqref{eq:GS}.

\begin{proposition}
\label{p:p4}
Suppose that $f\neq0$ and $g\neq0$. Then an equation from \eqref{eq:Lienard} is equivalent to \eqref{eq:Ince_V} via \eqref{eq:GS} if
\begin{equation}\label{eq:criterion_2}
\begin{gathered}
 \left( 4  l^{2}-2  lp+m \right) f^{6} \left[  \left( \mu+1\right) f^{6}+\mu   \left( 32  \mu+27 \right) l^{2}-2   \mu ^{2}\left( 2  lp-m \right)  \right] -4  l^{3} \left( \mu+1 \right) f^{9}-\\
 12  l\mu   \left[ \left( 6  l^{2}-2  lp+m \right)  \left( 2  lp-m \right) \mu+9  l^{4} \right] f^{3}+ \mu ^{3} \left(m -2  lp \right) ^{3}=0,
 \end{gathered}
\end{equation}
where $l$ is defined in \eqref{eq:l_def} and
\begin{equation}\label{eq:criterion_2_p}
\begin{gathered}
p=l+gf_{x}, \quad m=fgl_{x}-pl, \quad \mu=4\beta^{2}-1.
 \end{gathered}
\end{equation}
Any equation from \eqref{eq:Lienard} whose coefficients satisfy \eqref{eq:criterion_2} possesses a first integral
\begin{equation}\label{eq:integral_2}
  I=H_{11}H_{12}^{-1},
\end{equation}
an integrating factor
\begin{equation}\label{eq:M_2}
  M=\left(\frac{F_{x}}{G}\right)^{2}\left(H_{10}H_{11}H_{12}\right)^{-1},
\end{equation}
where $H_{10}$, $H_{11}$ and $H_{12}$ are invariant curves of equations from \eqref{eq:Lienard} defined by \eqref{eq:criterion_2}, which have the form
\begin{equation}\label{eq:Ic_2}
\begin{gathered}
H_{10}=\frac{F_{x}}{G}y+F^{2}-\beta^{2},\quad \lambda_{10}=G\\
H_{11}=(F-\beta)I_{2\beta}(2\sqrt{H_{10}})-\sqrt{H_{10}}I_{2\beta+1}(2\sqrt{H_{10}}), \quad \lambda_{11}=-FG,\\
H_{12}=(F-\beta)K_{2\beta}(2\sqrt{H_{10}})+\sqrt{H_{10}}K_{2\beta+1}(2\sqrt{H_{10}}), \quad \lambda_{11}=-FG.
\end{gathered}
\end{equation}

The functions that generate transformations \eqref{eq:GS} are given by correlations
\begin{equation}\label{eq:transformations_2a}
G={\frac { \left[\mu   \left(2  lp -4  l^{2}-m \right) f^{3} -2  l \left( \mu+1 \right) f^{6}-2  \mu l  \left( 4  \mu l  p+9  l^{2}-2  \mu m \right)  \right] f}{ \left( \mu+1 \right) f^{9}-\mu   \left(2  \mu l  p-16  \mu l^{2}-15  l^{2}-\mu m \right) f^{3}-3 \mu^{2}  l \left( 2  lp-m \right) }},
\end{equation}
\begin{equation}\label{eq:transformations_2b}
F=\frac{f+G}{2G}.
\end{equation}
\end{proposition}
\begin{proof}
The beginning of the proof is similar to those of Proposition \ref{p:p3}. After applying \eqref{eq:GS} to \eqref{eq:Ince_V} we obtain equation \eqref{eq:Lienard}, whose coefficients are given by
\begin{equation}\label{eq:p4_2}
h=\frac{GF_{xx}-G_{x}F_{x}}{GF_{x}}, \quad f=(2F-1)G, \quad g=\frac{G^{2}(F^{2}-\beta^{2})}{F_{x}}.
\end{equation}
Consequently, any equation that is equivalent to \eqref{eq:Ince_XV} via \eqref{eq:GS} is of the form \eqref{eq:Lienard}. Conversely, if for a given functions $f$, $g$ and $h$ there is a solution of \eqref{eq:p4_2} for $F$ and $G$ such that $GF_{x}\neq0$, then the corresponding equation from \eqref{eq:Lienard} is equivalent to \eqref{eq:Ince_V}.

Now we need to find integrability conditions for \eqref{eq:p4_2} as an overdetermined system of equations for the functions $F$ and $G$. Since $G\neq0$ we exclude the function $F$ from the system \eqref{eq:p4_2} and, after some simplifications and with the help of \eqref{eq:l_def} and denoting $\mu=4\beta^{2}-1$, we obtain
\begin{equation}\label{eq:p4_3}
\mu G^{3}+ f^{2}G+2l=0, \quad 2\mu  f (g G_{x}- fG^{2})+(2f^{3}-2\mu g f_{x}-2\mu l)G+4fl=0.
\end{equation}
Suppose that $\mu\neq0$. Then differentiating the first equation from \eqref{eq:p4_3} and employing the second one we obtain the expression
\begin{equation}\label{eq:p4_4}
2 \mu f \left( f^{3}+3  l \right)  G^{2}+2   \left( 4 \mu  f^{2}l -f^{5}\right) G +2   \left(3  glf_{x}+3  {l}^{2} -fgl_{x}\right) \mu-4  f^{3}l=0.
\end{equation}
With the help of the first equation from \eqref{eq:p4_3} and \eqref{eq:p4_4} we find expression \eqref{eq:transformations_2a} for the function $G$. Notice that the case $f^{3}+3  l$ leads to $F_{x}=0$ and, hence, is not
considered.

Then we substitute \eqref{eq:transformations_2a} into the first equation from \eqref{eq:p4_2}, which splits into two cases. The first one is an algebraic correlation on $f$, $l$ and $\mu$. If we proceed further with this correlation, we arrive at an algebraic equation for the parameter $\mu$. Since we suppose that $\mu$ is an arbitrary we do not take this case into consideration.

The second case is exactly correlation \eqref{eq:criterion_2}. To check that there are no new compatibility conditions appear we differentiate \eqref{eq:p4_4} two times. The second derivative of \eqref{eq:p4_4} vanishes if we take into account the first one. Furthermore, the first derivative of \eqref{eq:p4_4} vanishes on \eqref{eq:criterion_2}. If we substitute \eqref{eq:transformations_2a} into the second equation from \eqref{eq:p4_2} the result also vanishes on \eqref{eq:criterion_2}.

Finally, we proceed with the case of $\mu=0$. Assuming that $\mu=0$ in \eqref{eq:p4_3} we obtain the expression for $G$ and the compatibility condition
\begin{equation}\label{eq:p4_5}
G=-\frac{2l}{f^{2}}, \quad gf^3(fl_{x}-3lf_{x})+l^2(f^3-4l)=0.
\end{equation}
However these two correlations immediately follow form \eqref{eq:criterion_2} and \eqref{eq:transformations_2a} at $\mu=0$ if one takes into account notations \eqref{eq:criterion_2_p}. Consequently, it is not necessary to consider the case of $\mu=0$ separately.

Finally, we apply Proposition 1 from \cite{Sinelshchikov2023a} to \eqref{eq:Ince_V_IC}, \eqref{eq:Ince_V_FI} and \eqref{eq:Ince_V_M} to obtain the expressions for first integral, integrating factor and invariant curves given in \eqref{eq:integral_2}, \eqref{eq:M_2} and \eqref{eq:Ic_2}. This completes the proof.
\end{proof}

\section{Examples}
\label{s:ex}

In this section we consider applications of Propositions \ref{p:p3} and \ref{p:p4} for finding integrable differential equations and dynamical systems. We provide two examples of non-Liouvillian integrable dynamical systems: one is a cubic Kolmogorov system and the other is a Li\'enard oscillator with arbitrary degree.


\textbf{Example 1.}
Suppose we have the following cubic system of equations
\begin{equation}\label{eq:ex1_1}
x_{t}=x(a_{0}-z-a_{2}x-a_{3}x^{2}), \quad z_{t}=z(b_{0}-b_{1}z-b_{2}x-b_{3}x^{2})+c_{3}x^{3}.
\end{equation}
Here $a_{0}$, $a_{2}$, $a_{3}$, $b_{i}$, $i=\overline{0,3}$ and $c_{3}$ are arbitrary parameters. System \eqref{eq:ex1_1} is a generalization of the Kolmogorov system ($c_{3}=0$) and the competitive Lotka--Volterra system ($c_{3}=a_{3}=b_{3}=0$) (see, e.g. \cite{Hernandez1997,Gine2009,Libre2022,Sinelshchikov2023a}).

\begin{center}
\setlength{\extrarowheight}{5pt}
\begin{table}[!ht]
\centering
  \begin{tabular}{|c|l|}
  \hline
   & Parameters of \eqref{eq:ex1_1}\\ \hline
  1 & $ b_{3}=a_{3}=0, \quad b_{0}=a_{0}, \quad  b_{1}=1, \quad  b_{2}=a_{2} \quad a_{0}c_{3}\neq0$  \\ \hline
 3 & $a_{0}=a_{3}=b_{0}=b_{3}=0, \quad  b_{1}=4/3, \quad b_{2}=4a_{2}/3, \quad a_{2}c_{3}\neq0$  \\ \hline
  6 & $a_{0}=b_{0}=0, \quad  b_{1}=4/3, \quad b_{2}=5a_{2}/3, \quad a_{3}=3b_{3}/4, \quad c_{3}=-a_{2}b_{3}/4, \quad a_{2}b_{3}\neq0$  \\ \hline
 8 & $a_{0}=a_{2}=b_{0}=b_{2}=0, \quad b_{1}=5/3, \quad b_{3}=5a_{3}/3, \quad c_{3}a_{3}\neq0$  \\ \hline
\end{tabular}
\caption{Integrable cases of system \eqref{eq:ex1_1}.}
\label{tab:tab1}
\end{table}
\end{center}

If we exclude $z$ from \eqref{eq:ex1_1}, we obtain the following equation from \eqref{eq:Lienard}
\begin{equation}\label{eq:ex1_2}
\begin{gathered}
x_{tt}-\frac{b_{1}+1}{x}x_{t}^{2}-\left[(2a_{3}+b_{3}-2a_{3}b_{1})x^2+(a_{2}+b_{2}-2a_{2}b_{1})x+2b_{1}a_{0}-b_{0}\right]x_{t} -\\
-a_{3}(a_{3}b_{1}-b_{3})x^5+(c_{3}-2a_{2}a_{3}b_{1}+a_{2}b_{3}+a_{3}b_{2})x^4+\left(2a_{0}a_{3}b_{1}-a_{2}^2b_{1}-a_{0}b_{3}+\right. \\ \left. a_{2}b_{2}-a_{3}b_{0}\right)x^3+(2a_{0}a_{2}b_{1}-a_{0}b_{2}-a_{2}b_{0})x^2-a_{0}(a_{0}b_{1}-b_{0})x=0.
\end{gathered}
\end{equation}
Now we apply criterion \eqref{eq:criterion_1} to \eqref{eq:ex1_1} and find that the coefficients of \eqref{eq:ex1_2} satisfy condition \eqref{eq:criterion_1} in eight cases. However, each pair of cases $\{1,2\}$, $\{3,4\}$, $\{5,6\}$ and $\{7,8\}$  correspond to the same equation from \eqref{eq:ex1_2} up to the transformation $a_{0}\rightarrow -b_{0}$, $a_{2}\rightarrow -3b_{2}$, $b_{3}\rightarrow b_{3}/2$ and $b_{3}\rightarrow a_{3}/2$, respectively. What is more, the dynamical systems that are represented by each pair are connected either by linear or quadratic transformation of the dependent variable $z$. Therefore, in Table \ref{tab:tab1} we present only four integrable cases of \eqref{eq:ex1_1} that are distinct. Transformations that connect \eqref{eq:ex1_2} to \eqref{eq:Ince_XV} are given in Table \ref{tab:tab2}.

\begin{center}
\setlength{\extrarowheight}{5pt}
\begin{table}[!ht]
\centering
  \begin{tabular}{|c|l|}
  \hline
   & Transformations \eqref{eq:GS} and parameters of \eqref{eq:Ince_XV}\\ \hline
  1 & $F=x, \quad  G=a_{0}x, \quad \alpha=\left[c_{3}/(2a_{0}^{2})\right]^{\frac{1}{3}}$  \\ \hline
  3 & $F=x^{1/3}, \quad  G=-(a_{2}x^{4/3})/3, \quad \alpha=\left[(3c_{3})/(2a_{2}^{2})\right]^{\frac{1}{3}}$   \\ \hline
  6 & $F=x^{-2/3}, \quad  G=(b_{3}x^{4/3})/2, \quad \alpha=-\left[(2a_{2})/(3b_{3})\right]^{\frac{2}{3}}$   \\ \hline
  8 & $F=x^{-1/3}, \quad  G=(a_{3}x^{5/3})/3, \quad \alpha=-\left[3c_{3}/(2a_{3}^{2})\right]^{\frac{1}{3}}$ \\ \hline
\end{tabular}
\caption{Transformations \eqref{eq:GS} that connect integrable case of \eqref{eq:ex1_2} with \eqref{eq:Ince_XV}.}
\label{tab:tab2}
\end{table}
\end{center}

For instance, the dynamical system from the case 1 has the form
\begin{equation}\label{eq:ex1_3}
x_{t}=x(a_{0}-z-a_{2}x), \quad z_{t}=z(a_{0}-z-a_{2}x)+c_{3}x^{3},
\end{equation}
With the help of the formulas \eqref{eq:Ic_1}, \eqref{eq:integral_1} and \eqref{eq:M_1} we find its invariant curves
\begin{equation}\label{eq:ex1_4}
\begin{gathered}
H_{13}= (a_{2}x+z)\mbox{Ai}\{ V\}-(4a_{0}c_{3})^{\frac{1}{3}}x\mbox{Ai}^{'}\{V\}, \quad \lambda_{13}=a_{0}-\frac{3}{2}(z+a_{2}x),\\
H_{14}=(a_{2}x+z)\mbox{Bi}\{ V\}-(4a_{0}c_{3})^{\frac{1}{3}}x\mbox{Bi}^{'}\{V\}, \quad \lambda_{14}=a_{0}-\frac{3}{2}(z+a_{2}x),\\
H_{15}=x, \quad \lambda_{15}=a_{0}-(z+a_{2}x),
  \end{gathered}
\end{equation}
and first integral and integrating factor
\begin{equation}\label{eq:ex1_5}
\begin{gathered}
I=H_{13}/H_{14}, \quad M=(H_{13}H_{14})^{-1}.
  \end{gathered}
\end{equation}
In the formulas above we use the notation
\begin{equation}\label{eq:ex1_6}
V=\frac{2c_{3}x^{3}+(z+a_{2}x)^{2}}{(4a_{0}c_{3})^{\frac{2}{3}}x^{2}}.
\end{equation}

One can construct invariant curves, firs integrals and integrating factors for the dynamical systems corresponding others cases from Table \ref{tab:tab1}. Since transformations \eqref{eq:GS} with the functions given in Table \ref{tab:tab2} preserve Liouvillian (non)integrability we conclude that dynamical systems from \eqref{eq:ex1_1} with the parameters given in Table \ref{tab:tab1} are not Liouvillian integrable, but have a non-Liouvillain first integrals.

\textbf{Example 2.}
Consider the following family of Lienard equations with degrees of $f$ and $g$ that are $(2n-1,3n-1)$
\begin{equation}\label{eq:ex2_1}
x_{tt}+\left(a_{2}x^{2n-1}+a_{1}x^{n-1}+a_{0}\right)x_{t}+b_{3n-1}x^{3n-1}+\ldots b_{0}=0,
\end{equation}
where $a_{2}b_{3n-1}\neq0$, $a_{1}$, $a_{0}$, $b_{i}$, $j=0,\ldots,2n-2$ are arbitrary parameters and $n\in\mathbb{N}\setminus \{1\}$. Notice that we exclude the case of $n=1$ from the consideration because at $n=1$ \eqref{eq:ex2_1} is essentially \eqref{eq:Ince_V}.

One can show that  coefficients of \eqref{eq:ex2_1} satisfy criterion \eqref{eq:criterion_2} if
\begin{equation}\label{eq:ex2_2}
\begin{gathered}
a_{2}=\frac{n}{2},\quad  a_{1}=\sigma n, \quad a_{0}=0, \quad b_{3n-1}=-\frac{n}{8},\\
b_{2n-1}=-\frac{n}{8}(2\sigma+1), \quad b_{n-1}=\frac{n\beta^{2}}{2}-\frac{n(2\sigma+1)^2}{8}, \\
 b_{j}=0,\, j\neq3n-1,2n-1,n-1,
\end{gathered}
\end{equation}
where $\sigma$ and $\beta$ are arbitrary parameters.

The functions $F$ and $G$ that generate equivalence transformations are
\begin{equation}\label{eq:ex2_2a}
F=\frac{1}{2}\left(x^{n}+2\sigma+1\right), \quad G=\frac{n}{2}x^{n-1}.
\end{equation}

Consequently, we obtain that the following Li\'enard equation
\begin{equation}\label{eq:ex2_3}
x_{tt}+\frac{n}{2}x^{n-1}(x^{n}+2\sigma)x_{t}+\frac{n}{8}x^{n-1}\left(4\beta^{2}-(x^{n}+2\sigma+1)^{2}\right)=0,
\end{equation}
has invariant curves
\begin{equation}\label{eq:ex2_4}
\begin{gathered}
H_{15}=4y+(x^{n}+2\sigma+1)^{2}-4\beta^{2}, \quad \lambda_{15}=\frac{n}{2}x^{n-1},\\
H_{16}=\sqrt{H_{15}}I_{2\beta+1}(\sqrt{H_{15}})-(x^{n}+2\sigma -2\beta+1)I_{2\beta}(\sqrt{H_{15}}),  \\
\lambda_{16}=-\frac{n}{4}x^{n-1}(x^{n}+2\sigma+1),\\
H_{17}=\sqrt{H_{15}}K_{2\beta+1}(\sqrt{H_{15}})+(x^{n}+2\sigma -2\beta+1)K_{2\beta}(\sqrt{H_{15}}),\\
\lambda_{17}=-\frac{n}{4}x^{n-1}(x^{n}+2\sigma+1).
\end{gathered}
\end{equation}
Equation \eqref{eq:ex2_3} also possesses first integral and integrating factor
\begin{equation}\label{eq:ex2_5}
 I=H_{16}H_{17}^{-1}, \quad M=(H_{15}H_{16}H_{17})^{-1}.
\end{equation}

\begin{figure}
  \centering
  \includegraphics[width=0.5\textwidth]{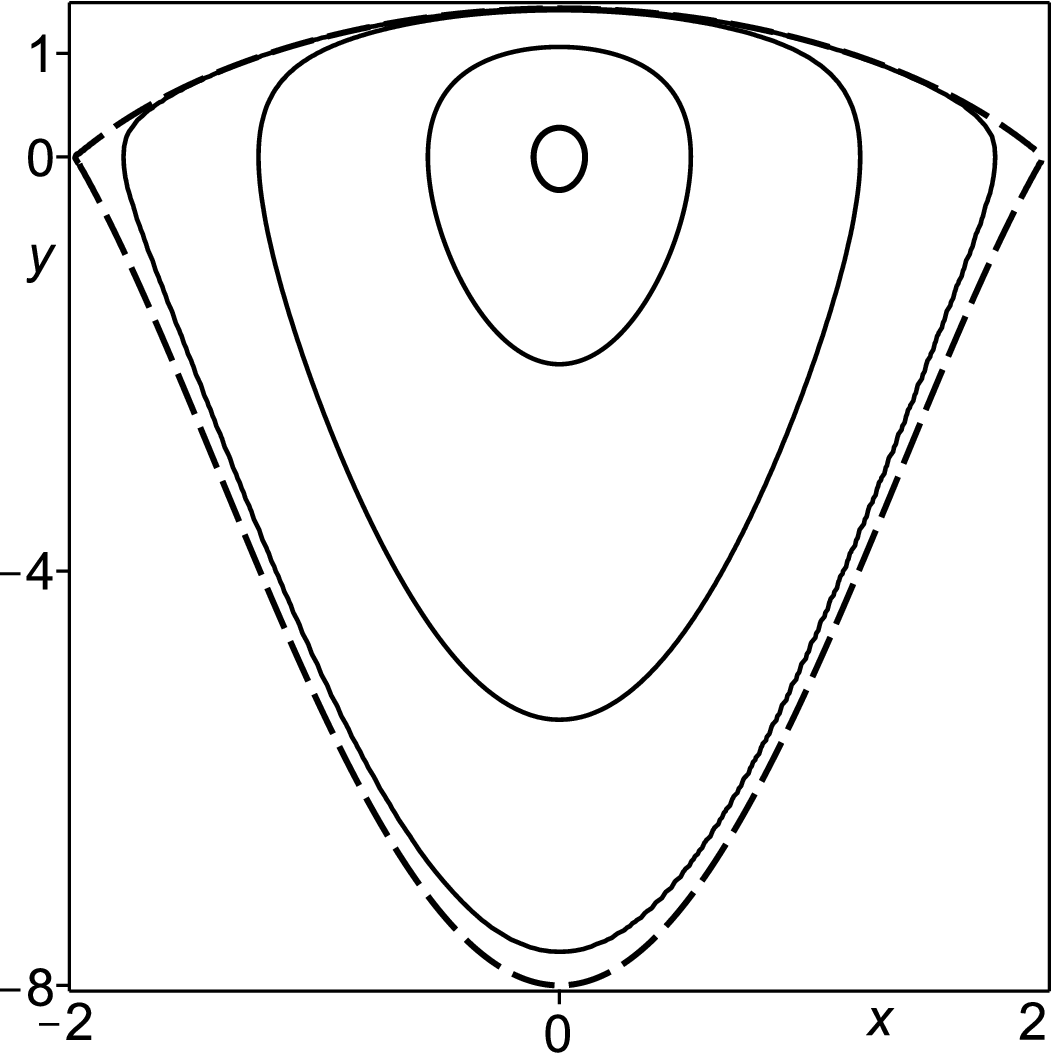}
  \caption{Level lines of first integral \eqref{eq:ex2_7a} (solid lines) and invariant curve \eqref{eq:ex2_7b} (dashed line) of \eqref{eq:ex2_6} at $\beta=1$, $\sigma=5/2$.}
  \label{fig:f1}
\end{figure}

Due to the fact that transformations \eqref{eq:GS} with \eqref{eq:ex2_2a} preserve Liouvillian (non)integrability we see that family of Li\'enard equations \eqref{eq:ex2_3} is integrable with non-Liouvillian first integrals and non-Puiseux integrating factor given in \eqref{eq:ex2_5}.

Let us consider a particular case of \eqref{eq:ex2_3} that corresponds to $n=2$. Then, from \eqref{eq:ex2_3} after applying scaling transformations $t=it^{'}$, $x=ix^{'}$ we obtain (the primes are omitted)
\begin{equation}\label{eq:ex2_6}
x_{tt}+(x^{2}+2\sigma)xx_{t}+\frac{x}{4}(4\beta^{2}-[x^{2}+2\sigma+1]^{2})=0,
\end{equation}
The first integral of \eqref{eq:ex2_6}, along with the corresponding invariant curves, can be obtained from \eqref{eq:ex2_5} by setting $x\rightarrow ix$. As a consequence, we get
\begin{gather}
 I=\widetilde{H}_{16}\widetilde{H}_{17}^{-1}, \label{eq:ex2_7a}\\
 \widetilde{H}_{15}=4y+(2\sigma+1-x^{2})^{2}-4\beta^{2} \label{eq:ex2_7d},\\
 \widetilde{H}_{16}=\sqrt{\widetilde{H}_{15}}I_{2\beta+1}\left(\sqrt{\widetilde{H}_{15}}\right)-(2\sigma-2\beta+1-x^{2})I_{2\beta}\left(\sqrt{\widetilde{H}_{15}}\right), \label{eq:ex2_7b}\\
 \widetilde{H}_{17}= \sqrt{\widetilde{H}_{15}}K_{2\beta+1}\left(\sqrt{\widetilde{H}_{15}}\right)+(2\sigma-2\beta+1-x^{2})K_{2\beta}\left(\sqrt{\widetilde{H}_{15}}\right).  \label{eq:ex2_7c}
\end{gather}
The dynamical system associated to \eqref{eq:ex2_6} has a center at the origin if $4\beta^{2}-(2\sigma+1)^{2}<0$. We demonstrate level curves $I=C_{1}$ of first integral \eqref{eq:ex2_7a} at $\beta=1$, $\sigma=5/2$ for several values of $C_{1}$ for $x\in(-2,2)$ (solid lines in Fig. \ref{fig:f1}). At these values of the parameters \eqref{eq:ex2_6} has a center at the origin, saddles at $(\pm 2,0)$ and unstable/stable nodes at $(\pm 2\sqrt{2},0)$. We also plot there the invariant curve \eqref{eq:ex2_7b} of \eqref{eq:ex2_6} (dashed line in Fig. \ref{fig:f1}). This invariant curve separates that regions of phase space with periodic dynamics and dynamics governed by saddles and nodes of \eqref{eq:ex2_6}.

\section{Conclusion}

In this work we have considered integrability for family of oscillators \eqref{eq:Lienard}. We have demonstrated that to construct non-Liouvillain and non-Puiseux first integrals and integrating factors one can use transcendental invariant curves. We have proposed an algorithm for finding and classifying such curves. Using this algorithm we have shown that two dynamical systems from the Painlev\'e--Gambier classification are not Liouvillian integrable and constructed their transcendental invariant curves and non-Liouvillian first integrals. It is also demonstrated that Puiseux integrabilty is not preserved by the algebraic transformations.
With the help of the proposed approach we have found an example of globally non-Puiseux integrable dynamical system that cannot be mapped into a Puiseux integrable one by algebraic transformations. In addition, we have constructed equivalence classes with respect to nonlocal transformations for Painlev\'e-type equations considered in this work. We have shown that there are interesting examples of non-Liouvilliant integrable dynamical systems among these equivalence classes. The algorithm proposed in this work can be applied not only to equations from family \eqref{eq:Lienard} but to arbitrary two-dimensional dynamical systems of the from \eqref{eq:gen_sys}. It is also interesting to consider extensions of this algorithm to higher dimensional dynamical systems. Finally, it is interesting to observe that when algorithms that use differential Galois theory are applied to study the (non)integrability of three-dimensional dynamical systems, the same linear equations for special functions (Bessel, hypergoemtric) appear as in this work, when we deal with transcendental invariant curves (see, e.g. \cite{Szuminski2020}). Moreover, singularities of solutions of these linear equations also play an important role in both approaches.

\section*{Acknowledgements.}
J.G. is partially supported by the Agencia Estatal de Investigac\'ion grant PID2020-113758GB-I00 and AGAUR grant number 2021SGR 01618. D.S. is partially supported by H2020-MSCA-COFUND-2020-101034228-WOLFRAM2. D.S. is also grateful to Ilia Gaiur for useful discussions.

\end{document}